\documentclass[accepted=2026-05-24,a4paper,onecolumn]{quantumarticle}
\pdfoutput=1

\usepackage[left=3cm,right=3cm,top=2cm,bottom=2cm]{geometry}

\usepackage[utf8]{inputenc}
\usepackage[T1]{fontenc}
\usepackage[english]{babel}

\usepackage{microtype}

\usepackage{authblk}

\usepackage{amsmath}
\usepackage{amsfonts}
\usepackage{mathdots} 

\usepackage[sorting=nyt, style=alphabetic, isbn=false]{biblatex}
\usepackage{csquotes}

\usepackage{hyperref}
\hypersetup{
    pdftitle={Exact distinguishability between real-valued and complex-valued Haar random quantum states},
    pdfauthor={Tristan Nemoz, Romain Alléaume, Peter Brown}
}
\usepackage[noabbrev]{cleveref}

\usepackage[hyperref, standard, thmmarks, amsmath]{ntheorem}
\usepackage{thmtools}
\usepackage{thm-restate}

\usepackage{tablefootnote}
\usepackage{algorithm}
\usepackage{algpseudocode}

\newcommand\symsubspacedt[1]{\vee^t\mathbb{#1}^d}
\newcommand\symsubspaced[2]{\vee^{#1}\mathbb{#2}^d}
\newcommand\symsubspace[3]{\vee^{#1}\mathbb{#2}^{#3}}
\newcommand{\proj}[1]{\Pi_{#1}}
\newcommand\projdtC{\Pi_{\symsubspacedt{C}}}
\newcommand\projdtR{\Pi_{\symsubspacedt{R}}}

\newcommand\projdR[1]{\Pi_{\symsubspaced{#1}{R}}}
\newcommand{\unitary}[1]{\mathcal{U}_{#1}}
\newcommand{\orthogonal}[1]{\mathcal{O}_{#1}}

\newcommand{\realvecs}[1]{\mathcal{R}_{#1}}
\newcommand{\complexvecs}[1]{\mathcal{C}_{#1}}
\newcommand\homogeneousdt{\mathcal{P}_d^{t}}
\newcommand\homogeneousd[1]{\mathcal{P}_d^{#1}}
\newcommand\homogeneous[2]{\mathcal{P}_{#1}^{#2}}
\newcommand\harmonicd[1]{\mathcal{H}_d^{#1}}
\newcommand\harmonic[2]{\mathcal{H}_{#1}^{#2}}

\newcommand{\meanrho}[1]{\rho_{#1}}
\newcommand{\normrrho}{\hat{\rho}_{\realvecs{d}^t}}
\newcommand{\crho}{\meanrho{\complexvecs{d}^t}}
\newcommand{\rrho}[2]{\meanrho{\realvecs{#1}^{#2}}}
\newcommand{\rrhod}[1]{\meanrho{\realvecs{d}^{#1}}}
\newcommand{\rrhodt}{\rrhod{t}}

\newcommand{\smallo}[1]{o\!\left(#1\right)}
\newcommand{\bigo}[1]{O\!\left(#1\right)}
\newcommand{\bigtheta}[1]{\Theta\!\left(#1\right)}

\newcommand{\expect}[1]{\mathbb{E}\left[#1\right]}

\newcommand{\ket}[1]{\left|#1\right\rangle}
\newcommand{\bra}[1]{\left\langle#1\right|}
\newcommand{\ketbra}[2]{\left|#1\middle\rangle\!\middle\langle#2\right|}
\newcommand{\selfouter}[1]{\left|#1\middle\rangle\!\middle\langle#1\right|}
\newcommand{\inner}[2]{\left\langle#1\middle|#2\right\rangle}
\newcommand{\braopket}[3]{\left\langle#1\middle|#2\middle|#3\right\rangle}
\newcommand{\polynomialbraket}[2]{\left\langle#1,#2\right\rangle}
\newcommand{\defed}{\overset{\text{def}}{=}}

\DeclareMathOperator{\trace}{tr}
\newcommand{\tr}[1]{\trace\left[#1\right]}
\DeclareMathOperator{\SpanOp}{span}
\newcommand{\Span}[1]{\SpanOp_{#1}}
\DeclareMathOperator{\rank}{rank}
\newcommand{\rk}[1]{\rank\left(#1\right)}
\newcommand{\NN}{\mathbb{N}}
\newcommand{\RR}{\mathbb{R}}
\newcommand{\CC}{\mathbb{C}}
\newcommand{\id}{\mathbb{I}}

\DeclareMathOperator{\Polylog}{polylog}

\DeclareMathOperator{\Negl}{negl}

\newtheorem*{theorem1}{Theorem 1}
\newtheorem*{corollary1}{Corollary 1}
\newtheorem*{proposition1}{Proposition 1}
\newtheorem*{proposition2}{Proposition 2}

\DeclareMathOperator{\generallinear}{GL}
\newcommand{\gl}[1]{\generallinear\left(#1\right)}
\DeclareMathOperator{\Hermitian}{Herm}
\newcommand{\Herm}[1]{\Hermitian\left(#1\right)}

\title{Exact distinguishability between real-valued and complex-valued Haar random quantum states}
\author{Tristan Nemoz}
\email{tristan.nemoz@telecom-paris.fr}
\orcid{0009-0004-3677-6094}
\author{Romain Alléaume}
\email{romain.alleaume@telecom-paris.fr}
\orcid{0000-0002-8322-3968}
\author{Peter Brown}
\orcid{0000-0001-9593-0136}
\email{peter.brown@telecom-paris.fr}
\affiliation{Télécom Paris, LTCI, Institut Polytechnique de Paris, Inria, 19 Place Marguerite Perey, 91~120 Palaiseau, France}

\begin{document}
	
	\maketitle
	\begin{abstract}
            \noindent Haar random states are fundamental objects in quantum information theory and quantum computing. 
            We study the density matrix resulting from sampling $t$ copies of a $d$-dimensional quantum state according to the Haar measure on the orthogonal group. In particular, we analytically compute its spectral decomposition. This allows us to compute exactly the trace distance between $t$-copies of a real Haar random state and $t$-copies of a complex Haar random state. Using this we show a lower-bound on the approximation parameter of real-valued state $t$-designs and improve the lower-bound on the number of copies required for imaginarity testing.
	\end{abstract}
	
	\section{Introduction}
            Haar random unitaries and states are important tools in the study of quantum information and quantum computing~\cite{Mel24}. They have found prolific applications to problems such as the development of shadow tomography~\cite{HKP20}, demonstrating quantum advantage~\cite{Nei+17} and benchmarking quantum computers~\cite{CBSNG19}. However, sampling a state from the Haar measure requires quantum circuits of exponential size~\cite{HP07}, which is impractical. Therefore, the consideration of practicality necessitates a search for ensembles of efficiently-preparable quantum states that, when sampled, closely resemble a Haar random state.
            
            Ensembles of random states that mimic Haar random states are more formally known as approximate state $t$-designs~\cite{AE07}. Roughly, an ensemble of states is an $\varepsilon$-approximate state $t$-design if when given $t$-copies of either a sample from the ensemble or a Haar random state, the probability we correctly distinguish the two is bounded above by $\frac12 + \varepsilon$. In other words, there is only an $\varepsilon$ advantage over randomly guessing. As such, any task requiring to sample \(t\) copies of a state according to the Haar measure on the unitary group can be performed up to precision \(\varepsilon\) by using an \(\varepsilon\)-approximate state \(t\)-design instead.
            
            Approximate state $t$-designs have recently found significant applications to quantum cryptography in the context of Asymptotically Random State generators (ARS)~\cite{BS19}, which are a particular cryptographic notion of a random state generator -- akin to a random number generator in classical cryptography. Loosely, an ARS is a sequence of $\varepsilon$-approximate state $t$-designs (indexed by the dimension of the states) with $t$ growing polynomially in the logarithm of the dimension and $\varepsilon$ remaining negligible. An ARS is thus a family of approximate state $t$-designs with $t$ growing modestly with the dimension. ARS have been critical in the development of constructions of pseudorandom states (PRS)~\cite{JLS18, BS19}, which are a quantum analogue of pseudorandom generators. For PRS, the information theoretic indistinguishability of ARS with Haar random states is replaced with a computational indistinguishability. That is, the ensembles cannot be distinguished from Haar random states by any computationally bounded agent. PRS are in turn widely studied in the quantum cryptography field as they not only represent a new approach for performing quantum cryptography, namely adding computational assumptions to surpass what's classically achievable, but because they also are a potentially weaker assumption than quantum-secure one-way functions~\cite{Kre21,KQST23}. This has led to many cryptographic primitives being built out of PRS in a black-box way, such as quantum bit commitments~\cite{AQY22} or multi-party computation~\cite{MY22} amongst others~\cite{MicrocryptZoo}.
            
            Interestingly, many ARS are constructed solely from real-valued quantum states~\cite{BS19,Aar+23,GB23,JMW24,CSBH25}. These constitute ensembles of real-valued quantum states for which sampling a polynomial number of copies of a state approximates sampling the same number of copies of a state from the Haar measure on the unitary group. Beyond cryptographic applications, such constructions have also been used to show that there is no efficient procedure to test whether a generic state requires complex numbers to be expressed in the standard formalism~\cite{HBK24}, this is quantified by the imaginarity of the state and has been studied in the context of quantum resource theories~\cite{HG18,Wu+21,Xue+21}.
        
            In order for an \(\varepsilon\)-approximate state \(t\)-design to capture well a Haar random sample, its approximation parameter \(\varepsilon\) should be as small as possible. Now, whilst real-valued quantum theory can sometimes be as powerful as its complex counterpart, for instance in computation~\cite{Shi03}, it has been shown that it is insufficient to describe all the phenomena we observe~\cite{Ren+21}. This leads to the very natural question
            \begin{quote}
            \begin{center}
                \emph{Are there fundamental limitations on the performance of real-valued state $t$-designs?}
            \end{center}
            \end{quote}
            If answered positively, this would also imply fundamental limitations on their various applications such as the performance of real-valued ARS.

    \subsection{Contributions}
	In this work we derive a fundamental lower bound on the approximation parameter \(\varepsilon\) of any real-valued \(\varepsilon\)-approximate state \(t\)-design. To achieve this result we analytically study the spectral decomposition of the density matrix resulting from sampling \(t\) copies of a \(d\)-dimensional real-valued quantum state according to the Haar distribution on the orthogonal group. This forms the main technical result of our work, based on representation theory of the orthogonal group and connections to Harmonic polynomials. We expect the proof technique and result (which are presented in \Cref{section:main}) to find further applications beyond this work.
    \begin{theorem1}[Spectral decomposition]
        Let \(d\geqslant2\) and \(t\geqslant1\) be two natural numbers. Let \(\rrhodt\) be the density matrix associated to sampling \(t\) copies of a real-valued state according to the Haar distribution of the orthogonal group, that is
        \begin{equation}
            \label{eq:first_mention_rrhodt}
            \rrhodt = \int_{O\in\orthogonal{d}}\left(O\selfouter{0}O^\top\right)^{\otimes t}\,\mathrm{d}\mu(O)\,.
        \end{equation}
        Then for \(k\in \{0,1, \dots, \left\lfloor\frac{t}{2}\right\rfloor-1\}\), the eigenvalues of \(\rrhodt\) are
        \begin{equation}
            \label{eq:first-mention-eigenvalues}
            \lambda_k=\frac{t!(d-2)!!}{(2k)!!(d+2t-2k-2)!!}
        \end{equation}
        with multiplicity
        \begin{equation}
            \alpha_k = \binom{d+t-2k-1}{d-1}-\binom{d+t-2k-3}{d-1}\,.
        \end{equation}
    \end{theorem1}
    Using this spectral decomposition, we are able to provide an exact analytical formula for the distinguishing advantage between real-valued Haar random states and complex-valued Haar random states. This completes a recently derived upper bound~\cite{Sch24}.
    \begin{corollary1}[Informal]
        Let $\rrhodt$ and $\crho$ be the density matrices associated with sampling $t$-copies of a real and complex Haar random state respectively. Then the maximal probability with which one can distinguish these states, assuming they are both chosen with probability \(\frac12\), is given by $\frac12 + \frac14 \left\| \rrhodt - \crho\right\|_1$ and
        \begin{equation}
            \left\|\rrhodt - \crho\right\|_1 = \sum_k \alpha_k \left|\lambda_k - \binom{d+t-1}{t}^{-1}\right|
        \end{equation}
        where $\alpha_k$ and $\lambda_k$ are given in the previous theorem. 
    \end{corollary1}

    In \Cref{section:applications} we describe several applications of this result. Firstly, by a data processing inequality for the trace-norm, we obtain a lower bound on the distinguishability between sampling $t$-copies of any real-valued ensemble of states and $t$-copies of a complex-valued Haar random state. This lower bound places a fundamental limitation on the performance of real-valued state $t$-designs. In turn, we also derive limitations on the performance of real-valued ARS, providing a bound on how $t$ scales in terms of $d$, limiting their performance. This is captured in the following proposition.
    \begin{proposition1}
        There exists a projective measurement that can distinguish any ensemble of real-valued states from the complex-valued Haar random ensemble with probability \(p\) if given \(t\sim\sqrt{2\ln\left(\frac{1}{2(1-p)}\right)d}\) copies.
    \end{proposition1}
    Finally, we can, using the result on the trace distance, generalize~\citeauthor{HBK24}'s result on the number of copies required for testing for imaginarity~\cite{HBK24}. In particular, we refine their lower-bound from \(\Omega\!\left(\sqrt{2^n}\right)\) to \(\sqrt{\ln\left(\frac32\right)2^n}+\smallo{\sqrt{2^n}}\). To achieve this, we perform an exact security analysis of the binary phases ARS. This exact analysis goes beyond the original asymptotic analyses of~\cite{BS19,AGQY22} and may be of independent interest. We then show that this bound can be improved to \(\sqrt{2\ln\left(\frac32\right)d}+\smallo{\sqrt{d}}\) for an arbitrary dimension \(d\). This not only generalizes their result to arbitrary dimensions, but improve the lower-bound by a factor \(\sqrt{2}\). Along the way, we also derive the distribution of imaginarity for Haar random states, showing they follow a power law. This is informally stated in the following Proposition.
    \begin{proposition2}[Informal]
        Any algorithm testing whether a quantum state has any amount of imaginarity requires at least \(\sqrt{2\ln\left(\frac32\right)d}+\smallo{\sqrt{d}}\) copies of that quantum state to succeed with probability at least \(\frac23\).
    \end{proposition2}

    \section{Notations}
        \label{section:notations}
	Let us start by introducing the relevant objects and notation that we will use. Let $d\geqslant2$ and \(t\geqslant1\) be two natural numbers. We use the notation \(\orthogonal{d}\) and \(\unitary{d}\) to denote the sets of \(d\times d\) orthogonal and unitary matrices respectively. If \(\mathcal{S}\) is a set of states or operators, we denote by \(\mathcal{S}^t\) the set of objects from \(\mathcal{S}\) to the $t$ tensor power, i.e.,
	\begin{equation}
		\mathcal{S}^t\defed \left\{x^{\otimes t}\middle|x\in\mathcal{S}\right\}\,.
	\end{equation}
	Two particularly important sets in this work are the set of real and complex-valued states, which we define as
	\begin{equation}
		\realvecs{d} = \left\{O\ket{0}\middle|O\in\orthogonal{d}\right\} \qquad \text{and} \qquad \complexvecs{d} = \left\{U\ket{0}\middle|U\in\unitary{d}\right\}\,,
	\end{equation}
	respectively and to which we implicitly attach their respective Haar measures. We will also make extensive use of the notation $[d]\defed\{0,\ldots,d-1\}$.

	Let $\mathfrak{S}_t$ be the symmetric group on $t$ objects. Let $\mathbb{K}$ be either $\RR$ or $\CC$. For any $\pi \in \mathfrak{S}_t$ we can construct the operator $P_\pi$ whose action on $(\mathbb{K}^d)^{\otimes t}$ is given by 
	\begin{equation}
		P_\pi(\ket{v_1} \otimes \dots \otimes \ket{v_t}) =\ket{v_{\pi^{-1}(1)}} \otimes \dots \otimes \ket{v_{\pi^{-1}(t)}}\,.
	\end{equation}
	We can then define the symmetric subspace as the subspace whose elements are invariant under these actions, i.e.
	\begin{equation}
		\symsubspacedt{K} \defed \left\{\ket{v} \in \left(\mathbb{K}^d\right)^{\otimes t} \, \middle| \, P_\pi \ket{v} = \ket{v} \,\, \forall \pi \in \mathfrak{S}_t\right\}\,.
	\end{equation}
    For a subspace \(F\) of a vector space \(E\), we denote by \(\proj{F}\) the projector onto \(F\), with \(E\) being implicit.

    We will also make extensive use of complexity notations throughout this work and so we provide the reader with a brief recall on the definitions of these notations. Though we are often working with functions of two parameters $t,d\in \mathbb{N}$, we will only be concerned with the setting where \(t\) is a function of \(d\) and so the functions that we consider are defined using only a single parameter.

        \begin{center}
        \begin{tabular}{|c|c|c|}
            \hline
            Notation & Definition & Intuitive meaning \\
            \hline
            \(f=\bigo{g}\) & \(\frac{f}{g}\) is bounded & \(g\) upper-bounds \(f\)\\
            \hline
            \(f=\smallo{g}\) & \(\frac{f}{g}\) converges to \(0\) & \(g\) \emph{strictly} upper-bounds \(f\)\\
            \hline
            \(f=\Omega(g)\) & \(g=\bigo{f}\) & \(g\) lower-bounds \(f\)\\
            \hline
            \(f=\omega(g)\) & \(g=\smallo{f}\) & \(g\) \emph{strictly} lower-bounds \(f\)\\
            \hline
            \(f=\bigtheta{g}\) & \(f=\bigo{g}\) and \(g=\bigo{f}\) & \(g\) and \(f\) have a similar asymptotic behavior\\
            \hline
            \(f\sim g\) & \(\frac{f}{g}\) converges to 1 & \(g\) and \(f\) have the same asymptotic behavior\\
            \hline
        \end{tabular}
        \end{center}
	
	\section{Eigenvalues of real Haar random states}
        \label{section:main}
	The goal of this section is to compute, for $d$ and $t$ being two natural numbers, the eigenvalues and associated multiplicities of the operator \(\rrhodt\) (see Theorem~\ref{thm:eigenvalues}), where
	\begin{equation}
		\rrhodt\defed\int_{O\in\orthogonal{d}}\left(O\selfouter{0}O^\top\right)^{\otimes t}\,\mathrm{d}\mu(O)\,.
	\end{equation}
	To begin, we will briefly review how the eigenvalues and multiplicities can be computed in the complex-case, i.e. for the operator
	\begin{equation}
		\crho\defed\int_{U\in\unitary{d}}\left(U\selfouter{0}U^\dagger\right)^{\otimes t}\,\mathrm{d}\mu(U)\,.
	\end{equation}
	We note that this case is well-studied and arguably much simpler. We refer the interested reader to~\cite{Har13} where many of the concepts in this introductory discussion can be found in significantly more detail.
	
	We begin by recalling two well-known facts about $\crho$. Firstly, the support of $\crho$ is contained entirely within the symmetric subspace $\symsubspacedt{C}$ and so for any $\ket{v} \in \symsubspacedt{C}$ we have $\crho \ket{v} \in \symsubspacedt{C}$. Secondly, for any $U \in \unitary{d}$ we have that
	\begin{equation}
		U^{\otimes t} \crho = \crho U^{\otimes t}\,.
	\end{equation}
    Using these two facts, we can apply some representation theory to diagonalize \(\crho\). We provide a brief introduction to necessary concepts in \Cref{app:rep-theory}. In particular, since \(U\mapsto U^{\otimes t}\) is an irreducible representation of \(\unitary{d}\) in \(\symsubspacedt{C}\)~\cite[Theorem~5]{Har13}, it follows by \Cref{cor:proportional-to-identity}, which is derived from \nameref{lem:schur}, that $\crho$ must be proportional to the identity operator on \(\symsubspacedt{C}\), namely $\projdtC$. That is, there exists \(\lambda\in\RR\) such that $\crho = \lambda\,\projdtC$. Finally, we know that $\tr{\crho}=1$ and $\tr{\projdtC} = \dim(\symsubspacedt{C}) = \binom{d+t-1}{t}$ and so
	\begin{equation}
		\crho = \binom{d+t-1}{t}^{-1}\,\projdtC\,.
	\end{equation}
	As $\projdtC$ is a projector we immediately obtain the spectral decomposition of $\crho$.
	
	We now proceed to try to obtain a similar result for the case of $\rrhodt$. We begin with two lemmas whose proofs can be found in \Cref{app:main_text_proofs}.
	\begin{restatable}{lemma}{onrealsymspace}
		\label{lem:on-real-symspace}
		Let $d\geqslant2$ and \(t\geqslant1\) be two natural numbers. The support of \(\rrhodt\) viewed as an endomorphism of \(\left(\RR^d\right)^{\otimes t}\) is contained within the real symmetric subspace \(\symsubspacedt{R}\).
	\end{restatable}
	\begin{restatable}{lemma}{commuteorthogonal}
		\label{lem:commute-orthogonal}
		Let $d\geqslant2$ and \(t\geqslant1\) be two natural numbers. For any \(O\in\orthogonal{d}\) we have
		\begin{equation}
			O^{\otimes t}\rrhodt=\rrhodt O^{\otimes t}\,.
		\end{equation}
	\end{restatable}
	Note at this point that the real-valued case appears to mirror the complex-valued case exactly. As such, we would now look to apply \Cref{cor:proportional-to-identity} as we did previously. Unfortunately this is not immediately possible as \(O\mapsto O^{\otimes t}\) is not an irreducible representation of $\orthogonal{d}$ in \(\symsubspacedt{R}\) for \(t>1\). Indeed, let us consider the case \(t=d=2\). We have
	\begin{equation}
		\label{eq:decompositiond=t=2}
		\symsubspace{2}{R}{2} = \Span{\RR}\left(\left\{\ket{01}+\ket{10},\ket{00}-\ket{11},\ket{00}+\ket{11}\right\}\right)\,.
	\end{equation}
	However, a quick calculation shows that for any $O \in \orthogonal{2}$ we have $O^{\otimes 2}(\ket{00} + \ket{11}) = \ket{00} + \ket{11}$. Hence $\Span{\RR}\left(\{\ket{00} + \ket{11}\}\right)$ is a one-dimensional $\orthogonal{2}$-invariant subspace and hence this representation is not irreducible.
Thus, in order to apply \Cref{cor:proportional-to-identity}, we first need to decompose \(\symsubspacedt{R}\) into subspaces in which the representation \(O\mapsto O^{\otimes t}\) of \(\orthogonal{d}\) is irreducible.
	
	In order to construct such a decomposition, it will be useful to demonstrate an isomorphism between \(\symsubspacedt{R}\) and the space of homogeneous polynomials with \(d\) variables of degree \(t\) which we denote as
	\begin{equation}
		\homogeneousdt\defed\Span{\RR}\left(\left\{X_0^{a_0}\cdots X_{d-1}^{a_{d-1}}\middle|\sum_{i=0}^{d-1}a_i=t\right\}\right)\,.
	\end{equation}
	The isomorphism can be shown to exist by noting that the two vector spaces are over the same field and have the same dimension, namely \(\binom{d+t-1}{t}\). However, it will prove useful to exhibit an explicit isomorphism between the two spaces. Consider the standard basis for $(\RR^d)^{\otimes t}$ which we may write as $\ket{x_1, \dots, x_t}$ for $x_i \in [d]$. To each basis vector we may associate a $d$-dimensional vector $a$ taking values in $[t]$ which counts the number of times each value in $[d]$ appears in the basis vector, i.e. $a_k = \left|\left\{ i\in[t] \, \middle| \, x_i = k\right\}\right|$. To construct a basis for $\symsubspacedt{R}$ we can apply the map $\ket{v} \mapsto \sum\limits_{\pi \in \mathfrak{S}_t} P_\pi \ket{v}$ to each basis vector in the standard basis -- the remaining unique vectors form the basis of $\symsubspacedt{R}$. We note that two vectors are the same under this mapping if and only if their associated $a$ vectors are the same. In particular, this implies that the basis for $\symsubspacedt{R}$ may be succinctly described by the unique $a$ vectors. We can then map these $a$ vectors to \(\homogeneousdt\) by interpreting the $a$ vector as encoding the the degrees of the different variables.
	
	To illustrate this isomorphism, let us take \(d=t=2\) once again. An orthonormal basis for \(\symsubspace{2}{R}{2}\) is
	\begin{equation}
		\left\{\ket{00}, \ket{11}, \frac{\ket{01}+\ket{10}}{\sqrt{2}}\right\}\,,
	\end{equation}
	whose elements have the respective $a$ vectors $\{(2,0), (0,2), (1,1)\}$. The associated orthonormal basis of \(\homogeneous{2}{2}\) would then be
	\begin{equation}
		\left\{X_0^2, X_1^2, \sqrt{2}\,X_0X_1\right\}
	\end{equation}
    where the norm of a monomial is given by the Bombieri norm~\cite{BBEM90}, which is defined as
    \begin{equation}
        \left\|\prod_i X_i^{a_i}\right\| \defed \sqrt{ \frac{\prod_i a_i !}{(\sum_i a_i)!} }\,.
    \end{equation} More formally, the isomorphism is given by
    \begin{equation}
        \label{eq:isomorphism}
       X_0^{a_0}\cdots X_{d-1}^{a_{d-1}}\mapsto \projdtR\ket{0^{a_0},\cdots,(d-1)^{a_{d-1}}}
    \end{equation}
    with \(\sum_ia_i=t\). In particular, this isomorphism described maps polynomial with unit Bombieri norm to quantum states with unit norm, as shown in~\Cref{lem:same_inner_product}.
    
	Now for $k\leq t$, consider the subspace of homogeneous harmonic polynomials defined as  
	\begin{equation}
		\harmonicd{k} \defed \left\{P\in\homogeneousd{k}\middle|\sum_{i=0}^{d-1}\frac{\partial^2P}{\partial X_i^2}=0\right\}\,.
	\end{equation}
	A classical result from representation theory of the orthogonal group~\cite[Theorem~2.12]{CW68} tells us that
	\begin{equation}
            \label{eq:decomposition-pdt}
		\homogeneousdt = \harmonicd{t}\oplus q\harmonicd{t-2}\oplus q^2\harmonicd{t-4}\oplus\cdots
	\end{equation}
	where
	\begin{equation}
		q\defed\sum_{i=0}^{d-1}X_i^2\,.
	\end{equation}
	Moreover, this result also states that there exists a certain representation of \(\orthogonal{d}\) that acts irreducibly on each of these $q^k \harmonicd{t-2k}$ subspaces. We show in \Cref{lem:representation-od} that this representation corresponds exactly to \(\orthogonal{d}\ni O\mapsto O^{\otimes t}\) when mapped back to \(\symsubspacedt{R}\) using the aforementioned isomorphism. As such, we can freely apply \Cref{cor:proportional-to-identity} on each of these subspaces. We further know from~\cite[Proposition~5.8]{ABW13} that for all \(k\in\left[1+\left\lfloor\frac{t}{2}\right\rfloor\right]\), we have
	\begin{equation}
            \label{eq:dim-subspaces}
		\dim\left(\harmonicd{t-2k}\right) = \binom{d+t-2k-1}{d-1}-\binom{d+t-2k-3}{d-1}\,.
	\end{equation}
	Note that in the case where \(d=2\) and \(t\) is even, we use the convention \(\binom{-1}{1}=0\), so that \(\dim\left(\harmonicd{0}\right)=1\). Thus, we know the multiplicities of the eigenvalues of \(\rrhodt\) and are only left with determining their values. In order to do so, we will explicitly construct corresponding eigenvectors.
	
	To achieve this we begin by constructing a more explicit expression for $\rrhodt$. For $x,y \in [d]^t$ and for $i \in [d]$, we define \(n_i(x, y)\) to be the number of times \(i\) appears in the concatenation of \(x\) and \(y\):
	\begin{equation}
		n_i(x, y)\overset{\text{def}}{=}\left|\left\{z\in x|z=i\right\}\right|+\left|\left\{z\in y|z=i\right\}\right|\,.
	\end{equation}
	This allows us to define the following equivalence relation on \([d]^t\):
	\begin{equation}
		x\sim y\iff \forall i\in[d], \,\, n_i(x,y)=0\!\!\!\!\pmod2\,.
	\end{equation}
	That is, \(x\) and \(y\) are equivalent if and only if no element of \([d]\) appears an odd number of times in the concatenation of \(x\) and \(y\). Using this relation we are able to give a more explicit form of $\rrhodt$.
	\begin{restatable}{lemma}{rhodtexpression}
		\label{lem:rhodtexpression}
		Let $d\geqslant2$ and \(t\geqslant1\) be two natural numbers. We have
		\begin{equation}
			\rrhodt = \frac{(d-2)!!}{(d+2t-2)!!}\sum_{x\sim y}\prod_{i=1}^d\left[n_i(x,y)-1\right]!!\ketbra{x}{y}
		\end{equation}
		with the convention \(-1!!=0!!=1\).
	\end{restatable}
	As the proof is mostly computational, we defer it to \Cref{app:proof_expression_rho}.
	
	With this explicit form we can use the isomorphism to generate eigenvectors $\ket{v_k}$ and compute analytically $\lambda$ such that $\rrhodt \ket{v_k} = \lambda \ket{v_k}$. Overall, we arrive at the following theorem whose proof is given in \Cref{app:proof_eigenvalues}.
    \begin{restatable}{theorem}{eigenvalues}
		\label{thm:eigenvalues}
        Let $d\geqslant2$ and \(t\geqslant1\) be two natural numbers and let $\{\lambda_k\}_k$ be the eigenvalues of $\rrhodt$ and let $\{\alpha_k\}_k$ be their respective multiplicities. For a homogeneous polynomial \(P\in\homogeneousdt\), we denote by \(\ket{P}\in\symsubspacedt{R}\) its image by the isomorphism described in~\Cref{eq:isomorphism}. Then for any \(k\in\left[1+\left\lfloor\frac{t}{2}\right\rfloor\right]\) and any \(P\in q^k\harmonicd{t-2k}\), \(\ket{P}\) is an eigenvector of \(\rrhodt\) with eigenvalue
        \begin{equation}
            \lambda_k = \frac{t!(d-2)!!}{(2k)!!(d+2t-2k-2)!!}\,.
        \end{equation}
        In particular, this eigenvalue has multiplicity
        \begin{equation}
            \alpha_k = \dim\left(q^k\harmonicd{t-2k}\right) = \binom{d+t-2k-1}{d-1}-\binom{d+t-2k-3}{d-1}\,.
        \end{equation}
        
	\end{restatable}
    \begin{example}
        If \(d=2\), then an orthonormal basis of \(\harmonic{2}{t}\) is
        \begin{equation}
            \label{eq:orthonormal_basis_h2t}
            \left\{\frac{1}{\sqrt{2^{t-1}}}\sum_{k=0}^{\left\lfloor\frac{t}{2}\right\rfloor}(-1)^k\binom{t}{2k}\,X_0^{2k}X_1^{t-2k},\frac{1}{\sqrt{2^{t-1}}}\sum_{k=0}^{\left\lfloor\frac{t-1}{2}\right\rfloor}(-1)^k\binom{t}{2k+1}\,X_0^{2k+1}X_1^{t-2k-1}\right\}\,.
        \end{equation}
        Note that multiplying these two polynomials by \(q^r\) for \(r\in\mathbb{N}\) yields an orthogonal basis of \(q^r\harmonicd{t}\), that we then have to renormalize in order to produce an orthonormal basis.
        We can then use this basis, together with \Cref{eq:decomposition-pdt}, to compute the spectral decomposition of \(\rrho{2}{t}\). For instance, let us take \(t=2\). An orthonormal basis of \(\harmonic{2}{2}\) is then
        \begin{equation}
            \left\{\frac{1}{\sqrt{2}}\,\left(X_1^2-X_0^2\right),\sqrt{2}\,X_0X_1\right\}
        \end{equation}
        which after applying the isomorphism results in the eigenvectors
        \begin{equation}
            \left\{\frac{\ket{11}-\ket{00}}{\sqrt{2}},\frac{\ket{01}+\ket{10}}{\sqrt{2}}\right\}\,.
        \end{equation}
        These eigenvectors are associated with the eigenvalue \(\frac14\). The remaining eigenvector is associated to the one-dimensional subspace $q$. This has an eigenvalue \(\frac12\) and after applying the isomorphism corresponds to \(\frac{\ket{00}+\ket{11}}{\sqrt{2}}\). All in all, we have
        \begin{equation}
                \rrho{2}{2} = \frac14\,\left(\proj{\frac{\ket{11}-\ket{00}}{\sqrt{2}}}+\proj{\frac{\ket{01}-\ket{10}}{\sqrt{2}}}\right)+\frac12\,\proj{\frac{\ket{11}+\ket{00}}{\sqrt{2}}}
        \end{equation}
        with \(\proj{\ket{\psi}}\) being \(\selfouter{\psi}\).
    \end{example}
    \begin{remark}[Constructing eigenvectors]
        For arbitrary $d$ and $t$ explicit formulae for bases of $\harmonicd{t}$ exist, see for example~\cite[Theorem~5.1]{DX13}. 
        In general, if \(\left\{Y_i\right\}_i\) is an orthonormal basis of spherical harmonics, then \(\left\{q^{\frac{t}{2}}Y_i\left(\frac{X}{\sqrt{q}}\right)\right\}_i\) is an orthonormal basis of \(\harmonicd{t}\).       
There are also efficient algorithms to generate an orthonormal basis of \(\harmonicd{t}\) for arbitrary \(d\) and \(t\) with software implementations available~\cite{HFT}. In~\Cref{app:example_spetrctal_decomposition} we provide a more detailed example of constructing the eigenvectors of \(\rrho{3}{3}\) and \(\rrho{3}{2}\) using the above ideas.
    \end{remark}
	
	As $\rrhodt$ and $\crho$ commute, computing the trace distance $\tfrac12 \|\rrhodt - \crho\|_1$ is now straightforward and we find that 
    \begin{equation}
        \frac12 \left\|\rrhodt - \crho\right\|_1 = \frac12 \sum_{k=0}^{\lfloor\frac{t}{2}\rfloor} \alpha_k \left|\lambda_k - \binom{d+t-1}{t}^{-1}\right| 
    \end{equation}
    where $\lambda_k$ and $\alpha_k$ are given in \Cref{thm:eigenvalues}. Through the operational interpretation of the trace distance in terms of distinguishability, this gives us precisely the advantage in distinguishing real-valued and complex-valued Haar random states when given access to $t$-copies. In certain cases we can also find a closed form expression for the trace distance. To begin, we note that $\rrhodt$ and $\crho$ have the same rank and trace. Therefore, as $\rrhodt$ always has more than one eigenvalue (when $d, t \geq 2$) we necessarily have that \(\lambda_0<\binom{d+t-1}{t}^{-1}\). The above trace distance simplifies in the case when we have that $\lambda_1 \geqslant \binom{d+t-1}{t}^{-1}$ which we describe in the following result.
	\begin{restatable}{corollary}{tracedistance}
		\label{cor:trace_distance}
		Let $d\geqslant2$ and \(t\geqslant1\) be two natural numbers such that
		\begin{equation}
                \label{eq:trace_distance_condition}
			\frac{t!(d-2)!!}{2(d+2t-4)!!}\geqslant \binom{d+t-1}{t}^{-1}\,.
		\end{equation}
		Then the trace distance between \(\crho\) and \(\rrhodt\) is given by
		\begin{equation}
                \label{eq:trace_distance}
			\frac12\left\|\crho-\rrhodt\right\|_1=\left(1-\frac{t(t-1)}{(d+t-1)(d+t-2)}\right)\left(1-\frac{(d+t-1)!}{(d-1)!!(d+2t-2)!!}\right)\,.
		\end{equation}
	\end{restatable}
        The proof of this corollary being only computational, we defer it to \Cref{app:trace_distance_proof}. We show in \Cref{lem:asymptotics} that if \(t<\frac{5+\sqrt{9+8d\ln\left(\frac{d}{2}\right)}}{2}\), then \Cref{eq:trace_distance_condition} holds for sufficiently large \(d\). In particular, this enforces \(t=\bigo{\sqrt{d\ln(d)}}\). Though this seems limiting, we can show that this covers every interesting regime for our applications, in that it covers the ones where the trace distance is negligible. Indeed, as shown in \Cref{lem:asymptotics}, under this condition, the trace distance can be written as
        \begin{equation}\label{eq:td-asymptotic}
            \frac12\left\|\crho-\rrhodt\right\|_1=\left(1-\bigtheta{\frac{t^2}{d^2}}\right)\left(1-\mathrm{e}^{-\frac{t(t-1)}{2d}+\bigtheta{\frac{t^3}{d^2}}}\right)\,.
        \end{equation}

        If we look more carefully, this splits into three asymptotic regimes wherein (i) the trace distance vanishes, (ii) the trace distance has a nontrivial lower bound and (iii) the trace distance converges to 1.
        \begin{enumerate}
            \item Let \(t=\smallo{\sqrt{d}}\), then we have
        \begin{equation}
            \frac12\left\|\crho-\rrhodt\right\|_1=\frac{t(t-1)}{2d}+\bigtheta{\frac{t^3}{d^2}}\,.
        \end{equation}
        Notably in this regime the trace distance vanishes asymptotically.
        \item If \(t\sim\alpha\sqrt{d}\), then the trace distance reduces to
        \begin{equation}
            \label{eq:trace-distance-converges-to-non-zero}
            \frac12\left\|\crho-\rrhodt\right\|_1=1-\mathrm{e}^{-\frac{\alpha^2}{2}}+\smallo{1}\,.
        \end{equation}
        In this case, we can see that the trace distance doesn't converge to 0.
        \item If \(t=\omega\!\left(\sqrt{d}\right)\), then the trace distance converges to 1.
        \end{enumerate}

        Finally, we note that the computation of the trace norm using Theorem~\ref{thm:eigenvalues} extends immediately to other Schatten $p$-norms, in particular we have that for any $p>1$
        \begin{equation}
            \left\|\rrhodt-\crho\right\|_{p} = \left(\sum_{k=0}^{\lfloor\frac{t}{2}\rfloor} \alpha_k \left|\lambda_k - \binom{d+t-1}{t}^{-1}\right|^p\right)^{1/p}
        \end{equation}
        and for the operator norm this simplifies to
        \begin{equation}
            \left\|\rrhodt-\crho\right\|_{\infty}=\max\left(\lambda_{\left\lfloor\frac{t}{2}\right\rfloor}-\binom{d+t-1}{t}^{-1},\binom{d+t-1}{t}^{-1}-\lambda_0\right)\,.
        \end{equation}

        \section{Applications}
        \label{section:applications}
        In this section, we use \Cref{cor:trace_distance} to show a lower-bound on the approximation parameter \(\varepsilon\) of real-valued \(\varepsilon\)-approximate state \(t\)-designs. We also refine and improve upon \citeauthor{HBK24}'s argument about the minimal number of queries required to test for the imaginarity of a state~\cite{HBK24}.
        
	\subsection{Real-valued approximate state \texorpdfstring{$t$}{t}-designs and ARS}
        Let us start by formally defining an \(\varepsilon\)-approximate state \(t\)-design.
        \begin{definition}[\(\varepsilon\)-approximate state \(t\)-design, adapted from {\cite{BS19}}]
            \label{def:state_designs}
            Let \(\varepsilon\in[0, 1]\) and \(t\in\mathbb{N}\). A set of quantum states \(\mathcal{S}\) equipped with a probability measure \(\lambda\) is an \(\varepsilon\)-approximate \(t\)-design if 
            \begin{equation}
                \frac12\left\|\int_{\ket{\psi}\in\mathcal{S}}(\selfouter{\psi})^{\otimes t}\,\mathrm{d}\lambda(\ket{\psi})-\crho\right\|_1\leqslant\varepsilon\,.
            \end{equation}
        \end{definition}
        
        By extension, we'll generally call an ensemble of states an \(\varepsilon\)-approximate state \(t\)-design when \(\varepsilon\) is the smallest number satisfying this definition, in which case we'll call \(\varepsilon\) the \emph{approximation parameter} of the approximate state design.
        
        Note that the definition usually restricts the approximate design to be a finite ensemble of quantum states. Here we allow for arbitrary ensembles. Furthermore, though this definition defines \(\varepsilon\)-approximate states \(t\)-designs using the trace norm, we could instead require closeness in spectral norm, as originally done in~\cite{AE07}. In all cases, the approximation parameter \(\varepsilon\) translates how well does the design reproduces the first \(t\) moments of the Haar measure. It is thus natural to ask whether there are fundamental lower-bounds on this approximation parameter given certain constraints.
        
        Following \citeauthor{BS19}'s footsteps, we can then use state $t$-designs to define Asymptotically Random States (ARS). 
        \begin{definition}[Asymptotically Random State, adapted from {\cite{BS19}}]
        An ARS, or Asymptotically Random State, is a sequence of \(\Negl(\log(d))\)-approximate state \(\log(d)^{\omega(1)}\)-designs, with \(d\) being the dimension of the quantum states and \(\Negl\) is a function such that for any polynomial $P$, \(\Negl(x)<\frac{1}{P(x)}\) for \(x\) being large enough.
        \end{definition}
        Whilst this definition is somewhat technical it can be readily unpacked with an example. 
        Consider the case of qubit systems, where the dimension is \(d=2^n\), with \(n\) being the number of qubits. Using this definition, an ARS is a sequence indexed by the number of qubits such that, when increasing the number of qubits by 1, one increases the number of copies (by more than \(1\)) that one can send to an adversary while retaining the same security.

        When analysing constructions of ARS, it is common to fix \(t\) and to compute the optimal approximation parameter as a function of the dimension \(d\) and the number of copies \(t\). As an example, let us consider the first ARS construction, introduced and studied in~\cite{JLS18}. For a fixed \(t\) and dimension \(d\), one considers the set of states
        \begin{equation}
            \mathcal{S}_{\text{Fourier}} = \left\{\ket{\psi_f^{\text{Fourier}}}\middle|f:\{0, 1\}^n\to\left[2^n\right]\right\}
        \end{equation}
        with \(\ket{\psi_f^{\text{Fourier}}}\) being defined as
        \begin{equation}
            \label{eq:prsfourier}
            \ket{\psi_f^{\text{Fourier}}}\defed\frac{1}{\sqrt{2^n}}\sum_{x\in\{0, 1\}^n}\mathrm{e}^{\frac{2\mathrm{i}\pi f(x)}{2^n}}\,\ket{x}\,.
        \end{equation}
        To make up an ensemble of states, we equip \(\mathcal{S}_{\text{Fourier}}\) with the uniform measure on the set of functions from the \(n\)-bit bitstrings \(\{0,1\}^n\) to \(\left[2^n\right]\). That is, sampling a state in \(\mathcal{S}_{\text{Fourier}}\) is done by assigning to the phase of each component \(\ket{x}\) of the statevector a randomly chosen power of the \(2^n\)-th root of unity \(\mathrm{e}^{\frac{2\mathrm{i}\pi}{2^n}}\).
        Let us denote \(\meanrho{\mathcal{S}^t_{\text{Fourier}}}\) the density matrix resulting from sampling \(t\) copies of a quantum state in \(\mathcal{S}_{\text{Fourier}}\) uniformly at random. The trace distance between \(\meanrho{\mathcal{S}^t_{\text{Fourier}}}\) and \(\crho\) has been exactly computed in~\cite[Lemma~1]{JLS18}. If we compute the asymptotics of this exact expression for the trace distance, we find that we have
        \begin{subequations}
            \begin{align}
                \frac12\left\|\meanrho{\mathcal{S}^t_{\text{Fourier}}}-\crho\right\|_1&=\frac{2^n\left(2^n-1\right)\cdots\left(2^n-t+1\right)}{2^{nt}}-\frac{2^n\left(2^n-1\right)\cdots\left(2^n-t+1\right)}{\left(2^n+t-1\right)\cdots2^n}\\
             &=\frac{t(t-1)}{2^{n+1}}+\bigtheta{\frac{t^4}{2^{2n}}}\,.
            \end{align}
        \end{subequations}
        Thus, \(\mathcal{S}_{\text{Fourier}}\) equipped with the uniform measure forms a \(\bigtheta{\frac{t^2}{2^n}}\)-approximate state \(t\)-design for all \(t\) and \(n\). Loosely speaking, increasing \(n\) by \(1\) allows to send \(\sqrt{2}\) more copies while keeping the same approximation parameter.
        
        As \(t\) grows, it becomes more difficult to reproduce the first $t$ moments of the Haar measure with a good precision, meaning that the optimal \(\varepsilon\) will at least grow with \(t\). For the purposes of applications, we are interested in understanding the speed of this growth. For example, consider the private-key quantum money scheme described in~\cite[Theorem~6]{JLS18}. For the sake of the argument, we will ignore all considerations relating to the complexity of producing the quantum states here. In this scheme, users can deposit a fixed amount of money in the bank, in exchange of which they receive a state \(\ket{\psi_{f_{\text{Bank}}}^{\text{Fourier}}}\), where \(f_\text{Bank}\) is a secret function, known only to the bank and identical for all users. Intuitively, the security of this scheme relies on the fact that users can't clone this quantum state, which means that the number of states that the bank can issue is limited, as it would otherwise risk the users colluding and perform tomography on this state. What the approximation parameter tells us is that if \(n\) is large enough, then as long as the bank issues less than \(\approx\sqrt{2^n}\) quantum states, the security of the system is ensured. Indeed, in that case, this quantum state cannot be readily distinguished from a Haar random state, even with \(\Omega(\sqrt{2^{n}})\) copies.

        Whilst this is not a formal security proof, morally the approximation parameter limits the performance of the ARS in practical applications. For instance, had the bank used an ARS with a smaller approximation parameter, such as \(\bigo{\frac{t}{2^n}}\) for the construction in~\cite{BS20}, then it could issue significantly more quantum states whilst remaining secure.

        It is thus of interest to understand any fundamental limitations that exist on the approximation parameter of state $t$-designs and in turn ARS. To this end, let us consider the following lemma, which shows such a limitation.

        \begin{lemma}[Distinguishing advantage lower bound]
		\label{lem:upper-bound-real-valued-ars}
		Let $d\geqslant2$ and \(t\geqslant1\) be integers, let \(p\in[1,\infty]\) and let \(\mathcal{S}_d\) be any given set of \(d\)-dimensional real-valued pure quantum states equipped with a probability measure \(\lambda\). Let 
        \begin{equation}
            \meanrho{\mathcal{S}_d^t} \defed \int_{\ket{\psi}\in\mathcal{S}_d}\left(\selfouter{\psi}\right)^{\otimes t}\,\mathrm{d}\lambda(\ket{\psi})\,,
        \end{equation}
        then we have
		\begin{equation}
			\left\|\meanrho{\mathcal{S}_d^t}-\crho\right\|_p\geqslant\left\|\rrhodt-\crho\right\|_p\,.
		\end{equation}
	\end{lemma}
	\begin{proof}
		Let us define \(\Phi\) to be the orthogonal group twirling channel:
		\begin{equation}
			\Phi\defed\rho\mapsto\int_{O\in\orthogonal{d}}O^{\otimes t}\rho\left(O^\top\right)^{\otimes t}\,\mathrm{d}\mu(O)\,.
		\end{equation}
		We have
		\begin{subequations}
			\begin{align}
				\Phi\left(\meanrho{\mathcal{S}_d^t}\right) &= \Phi\left(\int_{\ket{\psi}\in\mathcal{S}_d}\left(\selfouter{\psi}\right)^{\otimes t}\,\mathrm{d}\lambda(\ket{\psi})\right)\\
				&= \int_{\ket{\psi}\in\mathcal{S}_d}\Phi\left(\left(\selfouter{\psi}\right)^{\otimes t}\right)\,\mathrm{d}\lambda(\ket{\psi})\\
				&= \int_{\ket{\psi}\in\mathcal{S}_d}\rrhodt\,\mathrm{d}\lambda(\ket{\psi})\\
				&= \rrhodt
			\end{align}
		\end{subequations}
		where on the third line we used the fact that we have we have \(\Phi\left(\left(\selfouter{\psi}\right)^{\otimes t}\right)=\rrhodt\) for any real-valued state \(\ket{\psi}\in\realvecs{d}\). Similarly, we have \(\Phi\left(\crho\right)=\crho\) as $\crho$ is invariant under $t$-copy unitary conjugation. Now, since \(\Phi\) is a unital map, it is contractive under the \(p\)-norm~\cite[Theorem~II.4]{Per+06}, which yields
		\begin{equation}
			\left\|\meanrho{\mathcal{S}_d^t}-\crho\right\|_p\geqslant\left\|\Phi\left(\meanrho{\mathcal{S}_d^t}\right)-\Phi\left(\crho\right)\right\|_p=\left\|\rrhodt-\crho\right\|_p
		\end{equation}
		which concludes the proof.
	\end{proof}

        The above lemma shows that \emph{any} real-valued approximate state $t$-design is intrinsically limited in how well in can approximate the first $t$-moments of the Haar distribution on the whole space. It also formally shows, what one would intuitively expect, real Haar random states are the real-valued state $t$-design with the smallest approximation parameter.

        Whilst the actual probability of distinguishing a real-valued ensemble of states and $\crho$ will likely be much higher than the lower bound implied by \Cref{lem:upper-bound-real-valued-ars}, it is interesting to note that an explicit projective measurement achieving the lower bound can easily by defined. If we let $\{\Pi_+, \id - \Pi_+\}$ be the optimal measurement for distinguishing \(\rrhodt\) and \(\crho\) then it turns out that 
        \begin{equation}
            \frac12 \tr{\Pi_+ \meanrho{\mathcal{S}_d^t} } + \frac12 \tr{(\id-\Pi_+) \crho} = \frac12 + \frac14 \| \rrhodt - \crho\|_1\,.
        \end{equation}
        That is, the measurement $\{\Pi_+, \id-\Pi_+\}$ achieves precisely this lower bound as shown in the following proposition.        
        \begin{proposition}
            \label{prop:distinguishing_pvm}
            Let \(d\geqslant2\) and \(t\geqslant1\) be two natural numbers and and let \(\mathcal{S}_d\) be any given set of \(d\)-dimensional real-valued pure quantum states equipped with a probability measure \(\lambda\). Let us denote
            \begin{equation}
                \meanrho{\mathcal{S}_d^t}\defed\int_{\ket{\psi}\in\mathcal{S}_d}(\selfouter{\psi})^{\otimes t}\,\mathrm{d}\lambda(\ket{\psi})\,.
            \end{equation}
            Let \(\Pi_+\) be the projector onto the eigenspaces of \(\rrhodt\) associated to an eigenvalue larger than \(\binom{d+t-1}{t}^{-1}\). We then have
            \begin{equation}
                \tr{\Pi_+\meanrho{\mathcal{S}_d^t}} = \tr{\Pi_+\rrhodt}\,.
            \end{equation}
        \end{proposition}
        \begin{proof}
            Note that we have by linearity
            \begin{equation}
                \tr{\Pi_+\meanrho{\mathcal{S}_d^t}} = \int_{\ket{\psi}\in\mathcal{S}_d}\tr{\Pi_+\left(\selfouter{\psi}\right)^{\otimes t}}\,\mathrm{d}\lambda(\ket{\psi})\,.
            \end{equation}
            Now, recall that by construction, \(\Pi_+\) can be written as
            \begin{equation}
                \Pi_+=\sum_{k}\proj{q^k\harmonicd{t-2k}}
            \end{equation}
            where \(\proj{q^k\harmonicd{t-2k}}\) is the projector onto the subspace of \(\symsubspacedt{R}\) isomorphic to \(q^k\harmonicd{t-2k}\), and where \(k\) indexes the eigenspaces where \(\rrhodt\) has eigenvalues larger than \(\binom{d+t-1}{t}^{-1}\). Since the representation \(\orthogonal{d}\ni O\mapsto O^{\otimes t}\) is irreducible on these spaces, it means that \(\proj{q^k\harmonicd{t-2k}}\) and \(O^{\otimes t}\) commute for all \(O\in\orthogonal{d}\) and all indexes \(k\). As a consequence, \(\Pi_+\) commutes with every \(O^{\otimes t}\) for \(O\in\orthogonal{d}\). This means that for all \(\ket{\psi}\in\realvecs{d}\) we have
            \begin{subequations}
                \begin{align}
                    \tr{\Pi_+\left(\selfouter{\psi}\right)^{\otimes t}} &= \bra{\psi}^{\otimes t}\Pi_+\ket{\psi}^{\otimes t}\\
                    &= \bra{0}\left(O^\top\right)^{\otimes t}\Pi_+O^{\otimes t}\ket{0}\\
                    &= \bra{0}\Pi_+\ket{0}
                \end{align}
            \end{subequations}
            where \(O\) is any orthogonal matrix such that \(O\ket{0}=\ket{\psi}\), and where on the third line we used the fact that \(O^{\otimes t}\) and \(\Pi_+\) commute. In particular, this quantity is independent of the input state. By linearity, we thus get that for all \(\ket{\psi}\in\realvecs{d}\), we have
            \begin{equation}
                \tr{\Pi_+\left(\selfouter{\psi}\right)^{\otimes t}}=\tr{\Pi_+\rrhodt}
            \end{equation}
            which implies by linearity and the fact that \(\lambda\) is a probability measure that            \begin{equation}
                \tr{\Pi_+\meanrho{\mathcal{S}_d^t}} = \tr{\Pi_+\rrhodt}\,.
            \end{equation}
        \end{proof}
        Let us try to give some intuition as to why this result holds. In order to achieve the same probability of success, an adversary can apply the orthogonal twirling channel to the state they were given, and then apply the \(\left\{\Pi_+,\id-\Pi_+\right\}\) measurement. However, there is an additional property, as shown in the proof, that this measurement is invariant under orthogonal conjugation and hence the distinguishing probability remains the same whether or not the twirling channel is applied. 

        \subsubsection{On the compromise between efficient generation and good approximation}

        \Cref{lem:upper-bound-real-valued-ars} expresses the relatively intuitive idea that of all the probability measures on real-valued states, the Haar measure on the orthogonal group is the closest one to the Haar measure on the unitary group. In particular, no real-valued ARS can outperform the one resulting from sampling real-valued quantum states using the Haar measure on \(\orthogonal{d}\) in terms of approximation parameter.

        The practicality of an approximate state design is essentially a trade-off between the complexity of sampling the state
        and how well it approximates the first moments of the Haar distribution. A common trick to relax the constraint on complexity is to replace all instances of random functions in an ARS construction, e.g. \Cref{eq:prsfourier} by example, by quantum-secure pseudorandom ones. By definition, such a function can't be distinguished from a truly random one assuming bounded computational resources. The cost of preparing a quantum state in that case is then counted in terms of how many oracle calls to this quantum-secure pseudorandom function are needed. A brief overview of different existing constructions is shown in \Cref{tab:sota_prs}.

        \begin{table}[ht]
            \centering
            \small
            \begin{tabular}{|c|c|c|c|c|}
                \hline
                Short description & \begin{tabular}{c}Approximation\\parameter\end{tabular} & \begin{tabular}{c}Preparation\\cost\end{tabular} & Real? & Ref. \\
                \hline
                \begin{tabular}{c}Random roots of\\ unity phases \end{tabular} & \(\bigtheta{\frac{t^2}{2^n}}\) & \begin{tabular}{c}1  oracle call\end{tabular} & No & ~\cite{JLS18}\\
                \hline
                \begin{tabular}{c}Hamiltonian Phase States\\with \(m\) random angles\end{tabular} & \(\mathrm{e}^{2nt-\frac{m}{2t}}+\bigo{\frac{t^2}{2^n}}\) & \begin{tabular}{c}\(\lceil m/n \rceil\) layers of\\CNOTs and Zs\\ \end{tabular} & No & \cite{BHHP24}\\
                \hline
                \begin{tabular}{c}Random binary phases \end{tabular} & \(\bigtheta{\frac{t^2}{d}}\)\tablefootnote{The exact scaling is derived in~\Cref{app:trace-distance-binary-phases}, the previous proofs gave asymptotically optimal upper-bounds.} & \begin{tabular}{c}1 oracle call\end{tabular} & Yes & \cite{JLS18}\tablefootnote{Though the binary phase construction has been introduced in~\cite{JLS18}, its security has first been shown in~\cite{BS19}, and a simpler proof was later given in\cite{AQY22}.}\\
                \hline
                \begin{tabular}{c}Random binary phases\\applied
                in layers of blocks\\of size \(\xi\)\end{tabular} & \(\bigo{\frac{nt^2}{2^\xi\xi}}\) & \begin{tabular}{c}\(\bigo{\log(t)\log(\xi)}\)\\circuit depth\end{tabular} & Yes & \cite{CSBH25}\\
                \hline
                \begin{tabular}{c}Uniform superposition over \\
                random subset of size $K$ of \\
                computational basis\end{tabular} & \(\bigo{\frac{t^2}{d}+\frac{t}{\sqrt{K}}+\frac{tK}{d}}\) & \begin{tabular}{c}1 oracle call\end{tabular} & Yes & \cite{JMW24}\tablefootnote{This construction was also studied in~\cite{GB23} on qubits systems.}\\
                \hline
                \begin{tabular}{c}Random subset construction \\ with random binary phases\end{tabular} & \(\bigo{\frac{t^2}{K}}\) & \begin{tabular}{c}3 oracle calls\end{tabular} & Yes & \cite{Aar+23}\\
                \hline
                \begin{tabular}{c}PFC\tablefootnote{The PFC ensemble is comprised of operators that are the product of a permutation, a binary phase operator and a Clifford operator.} ensemble applied\\on a fiducial state \end{tabular} & \(\bigo{\frac{t}{\sqrt{d}}}\) & \begin{tabular}{c}\(\bigo{tn\Polylog(n)}\) \\ circuit depth\end{tabular} & No & \cite{MPSY24}\\
                \hline
                \begin{tabular}{c}\(\varepsilon\)-net with security\\parameter \(\lambda\) \end{tabular} & \(\bigo{\frac{t}{\mathrm{e}^{\lambda}}+\frac{\sqrt{t}+\lambda} {2^\lambda}+\left(\frac45\right)^\lambda}\) & \begin{tabular}{c}\(\bigtheta{n\lambda+\lambda^2}\)\\oracle calls\end{tabular} & No & \cite{BS20}\\
                \hline
                \begin{tabular}{c}Parallel Kac's walk with\\security parameter \(\lambda\)\end{tabular}& \(\bigo{\frac{\lambda nt}{\lambda^{\log(\lambda)}n^{\log(n)}}+\frac{t}{2^{\lambda n}}}\) & \begin{tabular}{c}\(\bigtheta{\lambda n}\)\\oracle calls\end{tabular} & No & \cite{LQSYZ24}\\
                \hline
            \end{tabular}
            \caption{Comparison of the approximation parameter and preparation cost of different ARS constructions. When available, the actual circuit depth is provided, and oracle calls are used otherwise, meaning that the actual depth depends on the implementation of the oracle.}
            \label{tab:sota_prs}
        \end{table}

        This table firstly shows that all of the real-valued constructions of ARS have an approximation parameter scaling like \(\Omega\!\left(\frac{t^2}{d}\right)\), which coincides with the results derived in \Cref{lem:upper-bound-real-valued-ars} and \Cref{cor:trace_distance}. Interestingly, they do so asymptotically optimally, in the sense that their approximation parameter scales quadratically with the dimension, matching the lower-bound derived from our result. Moreover, two constructions~\cite{BS20,LQSYZ24} have an approximation parameter that scales linearly with the number of copies and inversely linearly with the dimension. Both of which necessarily have high imaginarity (see \Cref{subsec:imaginarity}), as they mimic the Haar measure on the unitary group closely.

        Having a complex-valued construction is not sufficient to have such an approximation parameter as the other constructions~\cite{JLS18, BHHP24, MPSY24} demonstrate with their quadratic difference between $t$ and $d$. Finally, one can note that the constructions with the smallest approximation parameters are also those with the most costly state preparation procedures, as one may expect from the fact that preparing Haar random states requires circuits of exponential size~\cite{HP07}. Thus the search for constructions which are practical in both the approximation parameter and the preparation cost remains important.

    \subsection{Testing the imaginarity of a state}
        \label{subsec:imaginarity}
        A recent line of work studied the concept of \emph{imaginarity} of quantum states in the context of quantum resource theory~\cite{HG18, Wu+21, Wu+21b, Xue+21}. Even though it has been noted that this theory may have operational relevance~\cite{Wu+21}, its main motivation has originally been to study the role of complex Hilbert spaces in quantum theory. In particular, some works identified experiments demonstrating the insufficiency of real Hilbert spaces to describe quantum theory~\cite{Ren+21}. Though the existence of such work may suggest that testing whether a quantum state is real is easy, it has been shown in~\cite[Theorem~2]{HBK24} that such a task actually required a large number of copies of the input state. In the following, we will detail their method and improve upon their result. Along the way, we also provide an exact security analysis of the binary phase ARS construction which may be of independent interest.
        
        Consider the following task: a tester \(\mathcal{A}\) is given \(t\) copies of a pure state \(\ket{\psi}\) and must say whether the \emph{imaginarity} of \(\ket{\psi}\) is larger than a certain threshold \(\delta\) with constant probability. Following~\cite[Equation~3]{HBK24}, we define the imaginarity of a pure state \(\ket{\psi}\) as
        \begin{equation}
            \mathcal{I}(\ket{\psi})\defed1-\left|\inner{\psi}{\overline{\psi}}\right|^2.
        \end{equation}
        A formal definition of such a tester is given below.
        \begin{definition}[Imaginarity tester, adapted from {\cite[Definition~2]{HBK24}}] Let $\delta > \beta > 0$ and $t \in \mathbb{N}$. 
            An algorithm \(\mathcal{A}\) is called a $(\beta, \delta, t)$-imaginarity tester if for all states $\ket{\psi} \in \mathbb{C}^d$ it satisfies the following.
            \begin{enumerate}
                \item (Completeness) If \(\mathcal{I}(\ket{\psi})\leqslant\beta\), then \(\mathbb{P}\left[\mathcal{A}\left(\ket{\psi}^{\otimes t}\right)=1\right]\geqslant\frac23\).
                \item (Soundness) If \(\mathcal{I}(\ket{\psi})\geqslant\delta\), then \(\mathbb{P}\left[\mathcal{A}\left(\ket{\psi}^{\otimes t}\right)=1\right]\leqslant\frac13\).
            \end{enumerate}
        \end{definition}

        It was shown in~\cite[Lemma~2]{HBK24} that Haar random states have high imaginarity on average, namely \(1-\frac{2}{d+1}\), but distinguishing them from sets of real-valued states with non-negligible probability may still require \(\Omega\left(\sqrt{d}\right)\) copies. Since computing the imaginarity of a state is one way an adversary may try to distinguish sampling from a set of real-valued states and sampling from the Haar measure, it follows that computing the imaginarity on average must require at least as many copies as required to distinguish these two ensembles.

        To see this more explicitly, let us consider an imaginarity tester \(\mathcal{A}\) using \(t\) copies and let us consider a real-valued \(\bigo{\frac{t^x}{d}}\)-approximate state \(t\)-design \(\left\{\ket{\psi_f}\right\}_{f}\), with \(x>0\). We can use \(\mathcal{A}\) as a potential strategy for distinguishing between this approximate design and Haar random states. Namely, we use \(\mathcal{A}\) to test for the imaginarity of the states they've been given and if they have imaginarity larger than $\delta$, then they will conclude that they were given a Haar random state, and will otherwise conclude that they were given a state from the approximate design family. However, since the approximate design has an approximation parameter of \(\bigo{\frac{t^x}{d}}\), \(\mathcal{A}\) must necessarily have \(\Omega\left(d^{\frac1x}\right)\) copies to distinguish them with constant probability, thus proving that testing for imaginarity requires at least \(\Omega\left(d^{\frac1x}\right)\) copies. In particular, the smaller \(x\), the larger the number of copies.

        An interesting question is then what is the smallest \(x\) such that there exists a real-valued \(\bigo{\frac{t^x}{d}}\)-approximate \(t\)-design as this will give us the best bound on the number of copies required to test for imaginarity using this technique. In~\cite{HBK24} the authors used the $\varepsilon$-approximate state $t$-design which consists of the states
        \begin{equation}
            \ket{\psi_{f, S}}=\frac{1}{\sqrt{K}}\sum_{x\in S}(-1)^{f(x)}\,\ket{x}
        \end{equation}
        equipped with the uniform probability measure and with \(|S|=K\) being a fixed parameter and \(S\subseteq\{0, 1\}^n\), with \(n\) being the number of qubits. This construction was introduced and studied in~\cite{Aar+23}, where it was shown that \(\varepsilon=\bigo{\frac{t^2}{K}}\). In order to have the smallest \(\varepsilon\) possible, \citeauthor{HBK24} used \(K=2^n\), effectively reducing to the binary phase state approximate design which is comprised of states of the form
        \begin{equation}
            \ket{\psi_f}=\frac{1}{\sqrt{2^n}}\sum_x(-1)^{f(x)}\,\ket{x}
        \end{equation}
        also equipped with a uniform probability measure, which was introduced in \cite{JLS18} and studied in~\cite{BS19, AGQY22}. In particular, it was proved in~\cite{AGQY22} that this construction yields an \(\varepsilon\)-approximate state \(t\)-design with $ \varepsilon = \frac{t(3t-1)}{2^{n+1}}+\bigo{\frac{t^2}{2^{2n}}}$. However, by focusing on a particular construction,~\cite{HBK24} left open the possibility for better bounds. Using \Cref{cor:trace_distance}, we can now close this door and affirm that no better asymptotic lower-bound can be found using this proof technique. In fact, using \Cref{eq:trace-distance-converges-to-non-zero}, we can improve this lower-bound.

        We provide in \Cref{app:trace-distance-binary-phases} an exact security analysis for the binary phases ARS~\cite{JLS18}. In particular, we show in \Cref{lem:trace_distance_binary_phase_ars} that the trace distance between the density matrix \(\meanrho{\mathcal{B}_d^t}\) associated to this ARS and \(\crho\) is given by
        \begin{equation}
            \label{eq:trace_distance_binary_phase_ars}
            \frac12\left\|\meanrho{\mathcal{B}_d^t}-\crho\right\|_1 = 1-\frac{1}{\binom{d+t-1}{t}}\sum_{k=0}^{\left\lfloor\frac{t}{2}\right\rfloor}\binom{d}{t-2k}\,.
        \end{equation}
        From this it follows that if \(t\sim\alpha\sqrt{d}\) then the right-hand-side of \Cref{eq:trace_distance_binary_phase_ars} converges to \(1-\mathrm{e}^{-\alpha^2}\). In particular, in order for \(\mathcal{A}\) to satisfy the completeness property, they should be able to distinguish these two states with probability at least \(\frac23\). Since we know that the optimal probability of distinguishing these two states is \(\frac12+\frac12T\), where \(T\) is the aforementioned trace distance, it follows that one must have \(T\geqslant\frac13\), which gives us \(\alpha\geqslant\sqrt{\ln\left(\frac{3}{2}\right)d}\). As such, this study allows to derive the exact prefactor of \citeauthor{HBK24}'s bound, which was not known beforehand.

        But now, we can instead use the family of states defined by the Haar measure on the orthogonal group. As per \Cref{eq:trace-distance-converges-to-non-zero}, we know that if \(t\sim\alpha\sqrt{d}\), then the trace distance between the two relevant density matrices converges to \(1-\mathrm{e}^{-\frac{\alpha^2}{2}}\). By a similar argument, we can then conclude that the number of copies required by \(\mathcal{A}\) is at least \(t=\sqrt{2\ln\left(\frac32\right)d}+\smallo{\sqrt{d}}\). This thus improves upon \cite{HBK24}'s bound by a \(\sqrt{2}\) factor. Furthermore, we know that this lower-bound is asymptotically optimal as \(\delta\) goes to 0, as our goal in this case is to distinguish between \(\rrhodt\) and \(\crho\), which we know from \Cref{cor:trace_distance} can be done using \(t\sim\sqrt{2\ln\left(\frac32\right)d}\) copies.

        Finally, note the previous reasoning works in the limit of \(\delta\) being close to \(0\), since we want in that case to distinguish an ensemble of states that is close to being real-valued and the Haar random ensemble. \citeauthor{HBK24} proved that it was in fact valid for all \(\delta\) lower than \(1-\frac{n^2}{\sqrt{2^n}}\) (for \(n\geqslant80\)). By studying the distribution of imaginarity of Haar random states, we can increase the regime of \(\delta\) in which this result holds, be more precise in how \(\delta\) affects the number of copies and simplify \citeauthor{HBK24}'s proof, removing the reliance on L\'evy's lemma. Furthermore, our technique also works in any dimension, as opposed to only qubit-based systems.
        
        Note that this theorem also shows the optimality of this bound as \(\delta\) goes to 0, as the task at hand is in that case equivalent to distinguishing \(\rrhodt\) and \(\crho\). To the best of our knowledge, the only other known technique to perform such a task is the \(\bigo{d}\) algorithm of \citeauthor{HBK24}~\cite{HBK24} to compute the imaginarity. Our measurement thus represents a \(\bigo{\sqrt{d}}\) improvement.
        
        All of this is summarized and formally proved in the following proposition. 
        
        \begin{proposition}[Generalization of {\cite[Theorem~2]{HBK24}}]
            Any imaginarity tester of a quantum state \(\ket{\psi}\) requires at least \(t\sim\sqrt{2\ln\left(\frac{1}{\frac23+\delta^{\frac{d-1}{2}}}\right)d}\) copies of \(\ket{\psi}\) for any \(\delta<3^{-\frac{2}{d+1}}\) and \(\beta<\delta\).
        \end{proposition}
        \begin{proof}
            Let \(\mathcal{E}_{\mathcal{I}\geqslant\delta}\) be the ensemble of states having an imaginarity larger than or equal to \(\delta\) and let
            \begin{equation}
                \meanrho{\mathcal{E}^t_{\mathcal{I}\geqslant\delta}}\defed\frac{1}{\mu\left(\mathcal{E}_{\mathcal{I}\geqslant\delta}\right)}\int_{\ket{\psi}\in\mathcal{E}_{\mathcal{I}\geqslant\delta}}(\selfouter{\psi})^{\otimes t}\,\mathrm{d}\mu(\ket{\psi})
            \end{equation}
            be the average density matrix of this ensemble, with \(\mu\) being the Haar measure on \(\complexvecs{d}\). It was shown in~\cite[Equation~(S17)]{HBK24} that we have
            \begin{equation}
                \frac12\left\|\meanrho{\mathcal{E}^t_{\mathcal{I}\geqslant\delta}}-\crho\right\|_1\leqslant\mathbb{P}[\mathcal{I}(\ket{\psi})<\delta]\,.
            \end{equation}
        We then get by triangle inequality
        \begin{equation}
            \frac12\left\|\meanrho{\mathcal{E}_{\mathcal{I}\geqslant\delta}^t}-\rrhodt\right\|_1\leqslant\mathbb{P}\left[\mathcal{I}(\ket{\psi})\leqslant\delta\right]+\frac12\left\|\rrhodt-\crho\right\|_1\,.
        \end{equation}
        Now, we know that an algorithm that is able to test for imaginarity must be able to distinguish \(\meanrho{\mathcal{E}_{\mathcal{I}\geqslant\delta}^t}\) and \(\rrhodt\) with probability larger than \(\frac23\). But we also know that such an algorithm can only do so with probability \(\frac12+\frac14\left\|\meanrho{\mathcal{E}_{\mathcal{I}\geqslant\delta}^t}-\rrhodt\right\|_1\)\,. As such, in order for such an algorithm to exist, one must necessarily have
        \begin{equation}
            \frac12\left\|\meanrho{\mathcal{E}_{\mathcal{I}\geqslant\delta}^t}-\rrhodt\right\|_1\geqslant\frac13\,.
        \end{equation}
        This in turn implies
        \begin{equation}
            \mathbb{P}\left[\mathcal{I}(\ket{\psi})\leqslant\delta\right]+\frac12\left\|\rrhodt-\crho\right\|_1\geqslant\frac13\,.
        \end{equation}
        Now, we know from \Cref{lem:distribution-of-imaginarity} that
        \begin{equation}
            \mathbb{P}\left[\mathcal{I}(\ket{\psi})\leqslant\delta\right] = \delta^{\frac{d-1}{2}}\,.
        \end{equation}
        In particular, if we pick \(\delta<\frac{1}{3^{\frac{2}{d-1}}}\), this ensures
        \begin{equation}
            \mathbb{P}\left[\mathcal{I}(\ket{\psi})\leqslant\delta\right]<\frac13\,.
        \end{equation}
        One thus requires at least as many copies to test for imaginarity with such a \(\delta\) as one requires to have
        \begin{equation}
            \label{eq:requiredasymptoticsimaginarity}
            \frac12\left\|\rrhodt-\crho\right\|_1\geqslant\frac13-\mathbb{P}\left[\mathcal{I}(\ket{\psi})\leqslant\delta\right]\,.
        \end{equation}
        We then know from the asymptotics derived in \Cref{eq:trace-distance-converges-to-non-zero} that if one has
        \begin{equation}
            t\sim\sqrt{2\ln\left(\frac{1}{\frac23+\delta^{\frac{d-1}{2}}}\right)d}
        \end{equation}
        then \Cref{eq:requiredasymptoticsimaginarity} is satisfied, which provides a lower-bound on the number of copies required to test for imaginarity.
        \end{proof}
        Note that as \(\delta\) gets closer to \(3^{-\frac{2}{d+1}}\), the scaling of the lower-bound changes, indicating that it may not be useful in this regime. Indeed, if we take for instance \(\delta=\left(\mathrm{e}^{-\frac{1}{2d}}-\frac23\right)^{\frac{2}{d-1}}\), then the lower-bound simply becomes \(1+\smallo{1}\). This seems to indicate that improvements on this bound for high values of \(\delta\) are certainly possible. Note however that taking \(\delta=1-\Omega\!\left(\frac{1}{d}\right)\) keeps the same scaling, which improves upon \citeauthor{HBK24}'s previous \(1-\frac{\log_2^2(d)}{\sqrt{d}}\) bound for \(\delta\).

        \section{Conclusion}
        In this work, we proved a fundamental lower-bound on the approximation parameter of any real-valued approximate state $t$-design. We showed that this bound was tight and attained by the design resulting from sampling from the Haar measure on \(\orthogonal{d}\). 
        A notable implication of our work is a bound on the asymptotic approximation parameters of ARS generators, showing that having non-zero imaginarity is a requirement to have approximation parameters scaling linearly with the number of copies. As a second application, we also generalized a previous result of \cite{HBK24} on imaginarity testing to arbitrary dimensions and improved upon the number of copies required for imaginarity testing.

        On a technical level, our main result was the derivation of the complete spectral decomposition of $t$-copies of a real Haar random state. We derived this through a classical result on a representation of the orthogonal group on the space of homogeneous polynomials. In particular, we provided closed-form analytical expressions for the eigenvalues and multiplicities as well as an algorithm to compute the corresponding eigenvectors. We believe that this isomorphism between the symmetric subspace and homogeneous polynomials could prove useful in other contexts when exploring applications of real-valued Haar random states and we leave such an exploration to future work. The potential for future applications of the representation theory of harmonic polynomials to quantum information theory also renews interest in finding simple, explicit orthonormal bases of \(\harmonicd{t}\) with respect to the Bombieri inner product.

        However, our proofs prompt several further questions. Firstly, though we gave an explicit measurement that can be used to distinguish any real-valued approximate state design from a Haar random state, it is unclear if this measurement can be efficiently implemented. Possible directions to showing this would be to explore techniques like the Schur transform~\cite{HarrowPhD}. If the measurement is efficiently implementable, this implies that our results for ARS will also extend to pseudorandom quantum states~\cite{JLS18}, as we will have derived distinguishing lower bounds that are accessible to computationally bounded adversaries. In particular, this will also imply fundamental limitations on the performance of real-valued constructions of pseudorandom states.
        
        There are also several ways in which the study of the differences between the Haar measures on the orthogonal and unitary group could be extended. For instance, though we know that having oracular access to a random orthogonal matrix and having oracular access to a random unitary one are easily distinguishable~\cite{HBK24}, we may ask \emph{how} could we make them indistinguishable. This may for instance lead to stronger necessary conditions for a unitary ensemble to be a pseudorandom unitary ensemble~\cite{JLS18}. For instance, what happens if, in addition to a random orthogonal matrix, one applies a random unitary taken from a unitary \(2\)-design ensemble? That would raise the imaginarity of the resulting unitary, which makes~\citeauthor{HBK24}'s argument not work in this case. In fact we may even ask whether applying a diagonal unitary with uniformly random entries would be enough to get a good approximate unitary design.
        
        Similarly, in the state version of this question, how better can an approximate state design be if one applies the orthogonal twirling to a state having non-zero imaginarity? This line of study was initiated in~\cite{Sch24}, where it was shown that exact state designs couldn't be created in this way for \(t\geqslant4\), but the exact approximation parameter may be derived as a function of the imaginarity of the seed state. It would seem intuitive that having an imaginarity of \(1-\frac{2}{d+1}\) is optimal, and that having an imaginarity of 0 is the worst possible case. We believe that our techniques could help solving these questions, as the resulting density matrices will all be convex combinations of projectors onto the harmonic subspaces of \(\homogeneousdt\).

        Finally, of particular interest, would be a study of further necessary conditions and sufficient conditions for approximate state $t$-designs to have an approximation parameter scaling linearly in $t$. Though we've proved in this work that having complex amplitudes is necessary, we do not know how much imaginarity an approximate state design requires to have such an approximation parameter or what other properties are also necessary. A comprehensive understanding of such conditions would pave the way towards designing efficient and practical ARS.

        \section*{Acknowledgements}
    	The authors thank Saïd Ladjal, Bertrand Meyer and David Madore for fruitful discussions. The authors thank the anonymous referees of Quantum for their relevant remarks on our work. The authors acknowledge funding from FranceQCI, funded by the Digital Europe program, DIGITAL-2021-QCI-01, project no. 101091675 as well as funding from the European Union’s Horizon Europe research and innovation programme under the project “Quantum Secure Networks Partnership” (QSNP, grant agreement No. 101114043).
	
	\printbibliography
	
	\appendix

\crefalias{section}{appendix}
\crefalias{subsection}{appendix}

        \section{A very quick introduction to representation theory}
        \label{app:rep-theory}
        
        In this section, we introduce the representation theory concepts used in the main text. Let \(G\) be a group and let \(V\) be a vector space of finite dimension \(d\). A \emph{representation} of \(G\) in \(V\) is a group homomorphism \(\tau:G\to\gl{V}\), with \(\gl{V}\) being the general linear group associated to \(V\). That is, the mapping preserves the group operation, i.e.,
        \begin{equation}
            \tau(gh) = \tau(g)\tau(h)
        \end{equation}
        for all $g,h \in G$. A nontrivial subspace\footnote{A subspace \(W\subseteq V\) is called nontrivial if \(W\neq\{0\}\) and \(W\neq V\).} \(W\subset V\) is called \(G\)-invariant if for all \(\ket{w}\in W\) and all \(g\in G\) we have \(\tau(g)\ket{w}\in W\). If no nontrivial \(G\)-invariant subspaces exist, then the representation is called \emph{irreducible}. A fundamental result in representation theory is Schur's lemma, which we now state. Note that this is an edulcorated version of the original lemma that we've adapted to our needs.
        \begin{lemma}[Schur]
            \makeatletter\def\@currentlabelname{Schur's lemma}\makeatother
            \label{lem:schur}
            Let \(G\) be a group, \(V\) be a finite dimensional vector space and \(\tau:G\to\gl{V}\) be an irreducible representation of \(G\) in \(V\). Let \(\phi\) be an endomorphism of \(V\). If \(\phi\) and \(\tau(g)\) commute for all \(g\in G\), then \(\phi\) is either an automorphism or nil.
        \end{lemma}
        Schur's lemma is commonly used in the quantum information field by the means of the following corollary, that allows to characterize a density matrix invariant by a group representation.
        \begin{corollary}
            \label{cor:proportional-to-identity}
            Let \(G\) be a group, \(V\) be a finite dimensional vector space over \(\RR\) or \(\CC\) and \(\tau:G\to\gl{V}\) be an irreducible representation of \(G\) in \(V\). Let \(\rho\in\Herm{V}\) be such that \(\rho\) commutes with \(\tau(g)\) for all \(g\in G\). Then there exists a \(\lambda\in\RR\) such that \(\rho=\lambda\,\id_V\).
        \end{corollary}
        \begin{proof}
            Since \(\rho\) is Hermitian it admits at least one real eigenvalue \(\lambda\). Now, note that \(\rho-\lambda\,\id_V\) commutes with \(\tau(g)\) for all \(g\in G\). By \nameref{lem:schur}, \(\rho-\lambda\,\id_V\) is either an automorphism or is nil. Since it cannot be an automorphism by definition of \(\lambda\), it follows that \(\rho-\lambda\,\id_V\) is nil, which concludes the proof.
        \end{proof}

        \section{On the isomorphism between \texorpdfstring{\(\homogeneousdt\)}{the space of homogeneous polynomials} and \texorpdfstring{\(\symsubspacedt{R}\)}{the symmetric subspace}}
        The goal of this section is to translate the result of~\cite[Theorem~2.12]{CW68} from the space of homogeneous polynomials \(\homogeneousdt\) to the symmetric subspace \(\symsubspacedt{R}\). That is, we want to show that the representation of the orthogonal group on \(\homogeneousdt\) as considered in \cite{CW68} maps back to \(\orthogonal{d}\ni O\mapsto O^{\otimes t}\) when using the isomomorphism that we consider between these two spaces.
        
       Recall the isomorphism between $\homogeneousdt$ and $\symsubspacedt{R}$ is given by
       \begin{equation}
       X_0^{a_0}\cdots X_{d-1}^{a_{d-1}}\mapsto \projdtR\ket{0^{a_0},\cdots,(d-1)^{a_{d-1}}}\,.
    \end{equation}
        The inner product $\langle \cdot , \cdot \rangle_{\text{Bombieri}}$ on $\homogeneousdt$ is defined for two monomials $M, N$ as
        \begin{equation}
            \label{eq:bombieri_inner_product_definition}
            \langle M , N \rangle_{\text{Bombieri}} = 
            \begin{cases}
                0 & \quad \text{if } M\neq N \\
                \frac{1}{t!} \prod\limits_{i=0}^{d-1} a_i! & \quad\text{otherwise}
            \end{cases}
        \end{equation}
        where $M = X_0^{a_0}\cdots X_{d-1}^{a_{d-1}}$. For the remainder of this section, if \(P\in\homogeneousdt\) then \(\ket{P}\in\symsubspacedt{R}\) denotes its image by the isomorphism, the notation is motivated by the fact that this isomorphism preserves the inner product, as shown in the following Lemma.  
        \begin{lemma}
            \label{lem:same_inner_product}
            Let \(d\geqslant2\) and \(t\geqslant1\) be two natural numbers. 
             Let \(P\) and \(Q\) be two elements of \(\homogeneousdt\), then we have
            \begin{equation}
                \polynomialbraket{P}{Q}_{\text{Bombieri}} = \inner{P}{Q}\,.
            \end{equation}
            That is, the isomorphism described in \Cref{eq:isomorphism} preserves the inner product.
        \end{lemma}
        \begin{proof}
            By linearity it is sufficient to show that the inner products are equal when $P$ and $Q$ are monomials. Then let $P = X_0^{a_0}\cdots X_{d-1}^{a_{d-1}}$ and \(Q=X_0^{b_0}\cdots X_{d-1}^{b_{d-1}}\). Then 
            \begin{equation}
                \ket{P} = \frac{1}{t!}\sum_{\sigma\in\mathfrak{S}_t}A_\sigma\ket{0^{a_0},\cdots,(d-1)^{a_{d-1}}}
            \end{equation}
            and 
            \begin{equation}
                \ket{Q} = \frac{1}{t!}\sum_{\pi\in\mathfrak{S}_t}A_\pi\ket{0^{b_0},\cdots,(d-1)^{b_{d-1}}}
            \end{equation}
            where we've denoted
            \begin{equation}
                A_\pi(\ket{v_1} \otimes \dots \otimes \ket{v_t}) =\ket{v_{\pi^{-1}(1)}} \otimes \dots \otimes \ket{v_{\pi^{-1}(t)}}
            \end{equation}
            for any permutation \(\pi\in\mathfrak{S}_t\) and any basis state \(\ket{v_1,\cdots,v_t}\) of \(\left(\CC^{d}\right)^{\otimes t}\). We thus have
            \begin{equation}
                \inner{P}{Q} = \frac{1}{(t!)^2}\sum_{\sigma\in\mathfrak{S}_t}\sum_{\pi\in\mathfrak{S}_t}\braopket{0^{a_0},\cdots,(d-1)^{a_{d-1}}}{A_\sigma A_\pi}{0^{b_0},\cdots,(d-1)^{b_{d-1}}}\,.
            \end{equation}
            Note that whenever there is an index \(i\) such that \(a_i\neq b_i\), then the expectation value in the summation vanishes, as it can be seen as the inner product between two different computational basis states. This is consistent with the definition of the Bombieri inner product, which should be $0$ whenever the two monomials are not equal.
            
            Let us thus assume that \(a_i=b_i\) for all indices \(i\). The expectation value in the sums in now either equal to \(0\) or \(1\), depending on whether \(\sigma\) and \(\pi\) send \(\left(0^{a_0},\cdots,(d-1)^{a_{d-1}}\right)\) to the same tuple. Let us fix \(\sigma\in\mathfrak{S}_t\) and let us count the number of \(\pi\in\mathfrak{S}_t\) such that \(\sigma\left(0^{a_0},\cdots,(d-1)^{a_{d-1}}\right)=\pi\left(0^{a_0},\cdots,(d-1)^{a_{d-1}}\right)\). To perform this counting, note that for such a permutation, we can swap the placements of the different \(0\) elements and still end up with a permutation having this property. There are \(a_0!\) such permutations of the \(0\) elements. Similarly, there are \(a_1!\) permutations of the \(1\) elements. Choosing a valid \(\pi\) is equivalent to starting from \(\sigma\) and choosing a permutation of the \(i\) elements that leaves the result invariant for all \(i\) elements. All in all, this amounts to \(\prod_ia_i!\) permutations. As such, we have
            \begin{equation}
                \inner{P}{Q}=\frac{1}{(t!)^2}\sum_{\sigma\in\mathfrak{S}_t}\prod_{i=0}^{d-1}a_i!=\frac{1}{t!}\prod_{i=0}^{d-1}a_i!
            \end{equation}
            which is consistent with the Bombieri inner product definition.
        \end{proof}
        Since the symmetric subspace is intrinsically linked to the tensor product, it is natural to wonder how the tensor product is translated through the isomorphism between \(\homogeneousdt\) and \(\symsubspacedt{R}\). To that end, we start by showing the following Lemma that relates together projctors onto symmetric subspaces of different sizes.
        \begin{lemma}
            \label{lem:superprojector}
            Let \(d\geqslant2\), \(t\geqslant1\) and \(k\geqslant1\) be three natural numbers. Then
            \begin{equation}
                \projdR{t+k}\left(\projdR{t}\otimes\projdR{k}\right)=\projdR{t+k}\,.
            \end{equation}
        \end{lemma}
        \begin{proof}
            Note that \(\projdR{t}\otimes\projdR{k}\) is the projector onto \(\symsubspaced{t}{R}\otimes\symsubspaced{k}{R}\). As such, if we show that \(\symsubspaced{t+k}{R}\) is a subspace of \(\symsubspaced{t}{R}\otimes\symsubspaced{k}{R}\), the result would follow.

            Note that \(\symsubspaced{t}{R}\otimes\symsubspaced{k}{R}\) is the space of states that are invariant under operators of the form \(P_\pi\otimes P_{\sigma}\), with \(\pi\in\mathfrak{S}_t\) and \(\sigma\in\mathfrak{S}_k\). As a recall, if \(\tau\in\mathfrak{S}_n\) for some natural number \(n\), then \(P_\tau\) is defined on the basis states of \(\CC^{\otimes n}\) as
            \begin{equation}
                P_\tau\ket{x_1,\cdots,x_n} = \ket{x_{\tau^{-1}(1)},\cdots,x_{\tau^{-1}(n)}}\,.
            \end{equation}
            But now, note that we can see applying a permutation \(\pi\) on the first \(t\) registers and a permutation \(\sigma\) on the last \(k\) ones as a permutation on the whole \(t+k\) registers. In particular, every state of \(\symsubspaced{t+k}{R}\) is also left invariant by all operators of the form \(P_\pi\otimes P_\sigma\) for \(\pi\in\mathfrak{S}_t\) and \(\sigma\in\mathfrak{S}_k\), which shows that
            \begin{equation}
                \symsubspaced{t+k}{R}\subseteq\symsubspaced{t}{R}\otimes\symsubspaced{k}{R}\,.
            \end{equation}
            In particular, the projectors onto these spaces satisfy the desired relation.
        \end{proof}
        We can now leverage this Lemma to prove that, up to an additional symmetrization using the projector onto the symmetric subspace, the tensor product of states in \(\symsubspacedt{R}\) relates to the product of polynomials in \(\homogeneousdt\).
        \begin{corollary}
            \label{cor:representation-of-product}
            Let \(d\geqslant2\), \(t\geqslant1\) and \(k\geqslant1\) be three natural numbers. Let \(P\in\homogeneousdt\) and \(Q\in\homogeneousd{k}\). The image of \(PQ\) by the isomorphism described in \Cref{eq:isomorphism} is \(\projdR{t+k}(\ket{P}\otimes\ket{Q})\).
        \end{corollary}
        \begin{proof}
            Let \(\left\{P_i\right\}_i\) and \(\left\{Q_j\right\}_j\) be two orthogonal bases of \(\homogeneousdt\) and \(\homogeneousd{k}\) respectively, and let us write
            \begin{subequations}
                \begin{equation}
                    P = \sum_i\alpha_i\,P_i
                \end{equation}
                and
                \begin{equation}
                    Q = \sum_j\mu_j\,Q_j\,.
                \end{equation}
            \end{subequations}
            We have on the one hand
            \begin{equation}
                PQ = \sum_{i,j}\lambda_i\mu_j\,P_iQ_j
            \end{equation}
            and on the other hand
            \begin{equation}
                \projdR{t+k}(\ket{P}\otimes\ket{Q}) = \sum_{i,j}\lambda_i\mu_j\,\projdR{t+k}\ket{P_i, Q_j}\,.
            \end{equation}
            We thus only have to show the result for monomials. Let \(X_0^{a_0}\cdots X_{d-1}^{a_{d-1}}\in\homogeneousdt\) and \(X_0^{b_0}\cdots X_{d-1}^{b_{d-1}}\in\homogeneousd{k}\) be two monomials. We first have by the definition of the isomorphism between the symmetric subspace and the space of homogeneous polynomials
            \begin{equation}
                \begin{split}
                    &\projdR{t+k}\left(\ket{X_0^{a_0}\cdots X_{d-1}^{a_{d-1}}}\otimes\ket{X_0^{b_0}\cdots X_{d-1}^{b_{d-1}}}\right)\\={}& \projdR{t+k}\left(\projdR{t}\ket{0^{a_0},\cdots,(d-1)^{a_{d-1}}}\otimes\projdR{k}\ket{0^{b_0},\cdots,(d-1)^{b_{d-1}}}\right)\,.
                \end{split}
            \end{equation}
            We can now factor out the different tensor products, which gives us
            \begin{equation}
                \begin{split}
                    &\projdR{t+k}\left(\ket{X_0^{a_0}\cdots X_{d-1}^{a_{d-1}}}\otimes\ket{X_0^{b_0}\cdots X_{d-1}^{b_{d-1}}}\right)\\={}&\projdR{t+k}\left(\projdR{t}\otimes\projdR{k}\right)\ket{0^{a_0},\cdots,(d-1)^{a_{d-1}},0^{b_0},\cdots,(d-1)^{b_{d-1}}}\,.
                \end{split}
            \end{equation}
            We can now use \Cref{lem:superprojector} to remove the projector into \(\symsubspacedt{R}\otimes\symsubspaced{k}{R}\), which gives us
            \begin{equation}
                \begin{split}
                    &\projdR{t+k}\left(\ket{X_0^{a_0}\cdots X_{d-1}^{a_{d-1}}}\otimes\ket{X_0^{b_0}\cdots X_{d-1}^{b_{d-1}}}\right)\\
                    ={}&\projdR{t+k}\ket{0^{a_0},\cdots,(d-1)^{a_{d-1}},0^{b_0},\cdots,(d-1)^{b_{d-1}}}\,.
                \end{split}
            \end{equation}
            Finally, since \(\projdR{t+k}\) is invariant under all permutations of registers, we can shift around the different registers to get
            \begin{equation}
                \projdR{t+k}\left(\ket{X_0^{a_0}\cdots X_{d-1}^{a_{d-1}}}\otimes\ket{X_0^{b_0}\cdots X_{d-1}^{b_{d-1}}}\right)=\projdR{t+k}\ket{0^{a_0+b_0},\cdots,(d-1)^{a_{d-1}+b_{d-1}}}
            \end{equation}
            which is equal to \(\ket{X_0^{a_0+b_0}\cdots X_{d-1}^{a_{d-1}+b_{d-1}}}\) by definition.
        \end{proof}
        We will now slightly generalize this statement to products of polynomials with more than two terms.
        \begin{corollary}
            \label{cor:representation-of-n-product}
            Let \(d\geqslant2\) and \(N\geqslant2\) be two natural numbers. Let \(k_1\geqslant1,\cdots,k_N\geqslant1\) also be natural numbers, and let us denote \(K=\sum_{i=1}^Nk_i\). For \(i\in[N]\), let \(P_i\in\homogeneousd{k_i}\). The image of \(\prod_{i=0}^{N-1}P_i\) by the isomorphism described in \Cref{eq:isomorphism} is \(\projdR{K}\bigotimes_{i=0}^{N-1}\ket{P_i}\).
        \end{corollary}
        \begin{proof}
            Let us prove the result by induction on \(N\). The case \(N=2\) is true as per \Cref{cor:representation-of-product}. Let us assume that this statement is true for some \(N\geqslant2\), and let us show that it remains true for \(N+1\).

            By assumption, the image of \(\prod_{i=0}^{N-1}P_i\) is \(\projdR{K-k_{N+1}}\bigotimes_{i=0}^{N-1}\ket{P_i}\). By \Cref{cor:representation-of-product}, we thus have that the image of \(\prod_{i=0}^{N}P_i\) by the isomorphism is
            \begin{equation}
                \projdR{K}\left(\projdR{K-k_{N+1}}\bigotimes_{i=0}^{N-1}\ket{P_i}\otimes\ket{P_{N}}\right) = \projdR{K}\left(\projdR{K-k_{N+1}}\otimes\projdR{k_{N+1}}\right)\bigotimes_{i=0}^{N}\ket{P_i}
            \end{equation}
            which then gives us, using \Cref{lem:superprojector}, that the image of \(\prod_{i=0}^{N}P_i\) is
            \begin{equation}
                \projdR{K}\bigotimes_{i=0}^{N}\ket{P_i}
            \end{equation}
            which concludes the proof.
        \end{proof}
        We are now ready to prove the main result of this section. The only thing left to do is to introduce the representation of \(\orthogonal{d}\) in \(\homogeneousdt\) used in~\cite{CW68}.

        As a recall, a representation \(f\) of \(\orthogonal{d}\) takes as input an orthogonal matrix \(O\in\orthogonal{d}\) and returns an invertible endormorphism \(f(O)\) of \(\homogeneousdt\). In~\cite{CW68}, the representation \(f\) of \(\orthogonal{d}\) that is considered acts as
        \begin{equation}
            f(O)(P) = P\left(O^\top X\right)
        \end{equation}
        for all \(O\in\orthogonal{d}\) and all \(P\in\homogeneousdt\). Here \(P\left(O^\top X\right)\) is better understood when seeing \(P\) as a function of a real vector \(X\in\RR^d\): the image of \(P\) by \(f(O)\) is the polynomial that results from applying \(P\) onto the vector \(O^\top X\) instead of \(X\).

        As an example, let us take the case \(d=t=2\), and let us consider the orthogonal matrix \[O=\frac{1}{\sqrt{2}}\begin{pmatrix}1&1\\-1&1\end{pmatrix}\,.\]
        As such, we have
        \begin{equation}
            O^\top X=\frac{1}{\sqrt{2}}\begin{pmatrix}1&-1\\1&1\end{pmatrix}\begin{pmatrix}
                X_0\\X_1
            \end{pmatrix}=\frac{1}{\sqrt{2}}\begin{pmatrix}
                X_0-X_1\\X_0+X_1
            \end{pmatrix}
        \end{equation}
        \(P\left(O^\top X\right)\) is then the homogeneous polynomial that results from applying the original \(P\in\homogeneous{2}{2}\) to this new vector. We know that any such \(P\) can be written as
        \[P=a\,X_0^2+b\,X_0X_1+cX\,_1^2\]
        with \(a, b, c\in\RR\). We then have
        \begin{subequations}
            \begin{align}
                P\left(O^\top X\right) &= a\left[\frac{1}{\sqrt{2}}\left(X_0-X_1\right)\right]^2+b\left[\frac{1}{\sqrt{2}}\left(X_0-X_1\right)\frac{1}{\sqrt{2}}\left(X_0+X_1\right)\right]+c\left[\frac{1}{\sqrt{2}}\left(X_0+X_1\right)\right]^2\\
                &= \frac12\left(a+b+c\right)\,X_0^2+(c-a)\,X_0X_1+\frac12\left(a-b+c\right)\,X_1^2\,.
            \end{align}
        \end{subequations}
        Note that for any two orthogonal matrices \(O_1\) and \(O_2\), we have \(f\left(O_1\right)f\left(O_2\right)=f\left(O_1O_2\right)\). Furthermore, \(f\left(O_1\right)\) is an invertible endomorphism of \(\homogeneousdt\), meaning that \(f\left(O_1\right)\in\gl{\homogeneousdt}\) as expected.

        Thus, the only thing left to be shown is that the image of \(P\left(O^\top X\right)\) by the isomorphism is \(O^{\otimes t}\ket{P}\).
        \begin{lemma}
            \label{lem:representation-od}
            Let \(d\geqslant2\) and \(t\geqslant1\) be two natural numbers. Let \(P\in\homogeneousdt\) and \(O\in\orthogonal{d}\). The image of  \(P\left(O^\top X\right)\) by the isomorphism described in \Cref{eq:isomorphism} is \(O^{\otimes t}\ket{P}\).
        \end{lemma}
        \begin{proof}
            Let \(\left\{P_i\right\}_i\) be an orthonormal basis of \(\homogeneousdt\). Let us write \(P\) in this basis, which gives us
            \begin{equation}
                P=\sum_i\lambda_i\,P_i\,.
            \end{equation}
            We then have
            \begin{equation}
                P\left(O^\top X\right) = \sum_i\lambda_i\,P_i\left(O^\top X\right)\,.
            \end{equation}
            We thus only have to show the result on monomials. Let \(Q=X_0^{a_0}\cdots X_{d-1}^{a_{d-1}}\in\homogeneousdt\) be a monomial. We have by definition
            \begin{equation}
                Q\left(O^\top X\right) = \left(\sum_{i=0}^{d-1}\braopket{i}{O}{0}\,X_{i}\right)^{a_0}\cdots\left(\sum_{i=0}^{d-1}\braopket{i}{O}{d-1}\,X_{i}\right)^{a_{d-1}}=\prod_{j=0}^{d-1}\left(\sum_{i=0}^{d-1}\braopket{i}{O}{j}\,X_i\right)^{a_j}\,.
            \end{equation}
            Since this is a product of polynomials, we can use \Cref{cor:representation-of-n-product} to assert that
            \begin{equation}
                \ket{Q\left(O^\top X\right)} = \projdtR\bigotimes_{j=0}^{d-1}\left(\sum_{i=0}^{d-1}\braopket{i}{O}{j}\,\ket{i}\right)^{\otimes a_{j}}\,.
            \end{equation}
            Indeed, the image of any \(P^a\) by the isomorphism is \(\ket{P}^{\otimes a}\), which can also be seen as a consequence of \Cref{cor:representation-of-n-product}. Now, note that we have
            \begin{equation}
                \sum_{i=0}^{d-1}\braopket{i}{O}{j}\,\ket{i} = O\ket{j}
            \end{equation}
            for any index \(j\in[d]\). This gives us
            \begin{equation}
                \ket{Q\left(O^\top X\right)} = \projdtR\bigotimes_{j=0}^{d-1}\left(O\ket{j}\right)^{\otimes a_{j}}\,.
            \end{equation}
            We can now factor the orthogonal matrices out of the sum like so
            \begin{equation}
                \ket{Q\left(O^\top X\right)} = \projdtR O^{\otimes t}\bigotimes_{j=0}^{d-1}\ket{j}^{\otimes a_{j}}\,.
            \end{equation}
            Now, recall that \(\projdR{t}=\projdtC\) commute with any unitary matrix tensored \(t\) times. As such we have
            \begin{equation}
                \ket{Q\left(O^\top X\right)} = O^{\otimes t}\projdtR\bigotimes_{j=0}^{d-1}\ket{j}^{\otimes a_{j}}\,.
            \end{equation}
            Finally, by definition of the isomorphism, we have
            \begin{equation}
                \projdtR\bigotimes_{j=0}^{d-1}\ket{j}^{\otimes a_{j}} = \ket{Q}
            \end{equation}
            which concludes the proof.
        \end{proof}    

	\section{Proofs from the main text}
        In this section we'll prove all the results that have been deferred from the main text. Since \nameref{lem:schur} will give us the spectral decomposition of \(\rrhodt\) on \(\symsubspacedt{R}\), we first show that \(\rrhodt\) is nil elsewhere.
        
        \label{app:main_text_proofs}
        \subsection{Proof of Lemma~\ref{lem:on-real-symspace}}
        \onrealsymspace*
        \begin{proof}
        For a permutation \(\pi\in\mathfrak{S}_t\), let us recall the definition of \(P_\pi\) on the basis states:
        \begin{equation}
                \forall\left(x_1,\cdots,x_t\right)\in[d]^t,P_\pi\ket{x_1,\cdots,x_t}=\ket{x_{\pi^{-1}(1)},\cdots,x_{\pi^{-1}(t)}}\,.
            \end{equation}
        In particular, note that for all \(\pi\in\mathfrak{S}_t\) and all \(\ket{\psi}\in\RR^d\) we have
        \begin{equation}
                P_\pi\ket{\psi}^{\otimes t}=\ket{\psi}^{\otimes t}\,.
            \end{equation}
        For all \(\pi\in\mathfrak{S}_t\), we then have
        \begin{subequations}
                \begin{align}
                        P_\pi\rrhodt &= P_\pi\int_{O\in\orthogonal{d}}\left(O\selfouter{0}O^\dagger\right)^{\otimes t}\,\mathrm{d}\mu(O)\\
                        &= \int_{O\in\orthogonal{d}}P_\pi\left(O\ket{0}\right)^{\otimes t}\left(\bra{0}O^\dagger\right)^{\otimes t}\,\mathrm{d}\mu(O)\\
                        &= \int_{O\in\orthogonal{d}}\left(O\ket{0}\right)^{\otimes t}\left(\bra{0}O^\dagger\right)^{\otimes t}\,\mathrm{d}\mu(O)\\
                        &= \rrhodt\,.
                    \end{align}
            \end{subequations}
        And similarly we have \(\rrhodt P_\pi=\rrhodt\). As shown in \cite[Proposition~1]{Har13}, we have
        \begin{equation}
                \projdtR=\frac{1}{t!}\sum_{\pi\in\mathfrak{S}_t}P_\pi\,,
            \end{equation}
        which combined with the above implies that $\projdtR \rrhodt = \rrhodt = \rrhodt \projdtR$ and so
        \begin{equation}
            \rrhodt = \projdtR \rrhodt \projdtR\,.
        \end{equation}
        However the right-hand-side is an operator whose support is, by construction, contained within $\symsubspacedt{R}$.
    \end{proof}
    \subsection{Proof of Lemma~\ref{lem:commute-orthogonal}}
    Now, in order to use \nameref{lem:schur} and in particular \Cref{cor:proportional-to-identity}, we show that \(\rrhodt\) commutes with every element of a given representation, here \(\orthogonal{d}\ni O\mapsto O^{\otimes t}\).
    \commuteorthogonal*
    \begin{proof}
        Let \(O\in\orthogonal{d}\). We have
        \begin{subequations}
        \begin{align}
                O^{\otimes t}\rrhodt&=\int_{P\in\orthogonal{d}}(OP\ket{0})^{\otimes t}\left(\bra{0}P^\dagger\right)^{\otimes t}\,\mathrm{d}\mu(P)\\
                &=\int_{Q\in\orthogonal{d}}(Q\ket{0})^{\otimes t}\left(\bra{0}Q^\dagger O\right)^{\otimes t}\,\mathrm{d}\mu\left(O^\dagger Q\right)\\
                &=\int_{Q\in\orthogonal{d}}(Q\ket{0})^{\otimes t}\left(\bra{0}Q^\dagger O\right)^{\otimes t}\,\mathrm{d}\mu\left(Q\right)\\
                &=\int_{Q\in\orthogonal{d}}(Q\ket{0})^{\otimes t}\left(\bra{0}Q^\dagger\right)^{\otimes t}\,\mathrm{d}\mu\left(Q\right)O^{\otimes t}\\
                &=\rrhodt O^{\otimes t}
            \end{align}
        \end{subequations}
        where on the second line we substituted \(Q=OP\) and on the third line we used the translation invariance of the Haar measure, namely that for any subset \(\mathcal{S}\) of \(\orthogonal{d}\) and any orthogonal matrix \(O\) we have \(\mu(O\mathcal{S})=\mu(\mathcal{S})\).
    \end{proof}
    \subsection{Proof of Lemma~\ref{lem:rhodtexpression}}
        Now that \Cref{cor:proportional-to-identity} and \cite[Theorem~2.12]{CW68} gave us the eigenvectors of \(\rrhodt\), we want to find a more convenient expression of \(\rrhodt\) to leverage this result. The following Theorem gives the exact expression of all of \(\rrhodt\)'s coefficients.
	\label{app:proof_expression_rho}
	\rhodtexpression*
	\begin{proof}
            Since \(O\) is sampled from the Haar measure, we can write
            \begin{equation}
                O\ket{0} = \sum_{x_1=0}^{d-1}\frac{W_{x_1}}{\|W\|}\ket{x_1}
            \end{equation}
            with \(W\in\RR^d\) being distributed according to a standard centered multivariate normal distribution \(\mathcal{N}\left(0,\id_d\right)\), as per \cite[Theorem~1]{Mez07}. Note that we have
            \begin{equation}
                (O\ket{0})^{\otimes t}=\sum_{x_1=0}^{d-1}\cdots\sum_{x_t=0}^{d-1}\prod_{k=1}^t\frac{W_{x_k}}{\|W\|}\,\ket{x_1,\cdots,x_t}
            \end{equation}
            which we can write as
            \begin{equation}
                (O\ket{0})^{\otimes t}=\sum_{x\in[d]^t}\prod_{k=1}^t\frac{W_{x_k}}{\|W\|}\,\ket{x}\,.
            \end{equation}
            We thus have
            \begin{equation}
                \left(O\selfouter{0}O^\top\right)^{\otimes t} = \sum_{x,y\in[d]^t}\prod_{k=1}^t\frac{W_{x_k}}{\|W\|}\prod_{i=1}^t\frac{W_{y_i}}{\|W\|}\,\ketbra{x}{y}
            \end{equation}
            which we can rewrite as
            \begin{equation}
                \left(O\selfouter{0}O^\top\right)^{\otimes t} = \sum_{x,y\in[d]^t}\prod_{k=1}^t\frac{W_{x_k}W_{y_k}}{\|W\|^2}\,\ketbra{x}{y}\,.
            \end{equation}
            This then gives us
		\begin{align}
			\rrhodt &= \underset{O\in\orthogonal{d}}{\mathbb{E}}\left[\left(O\selfouter{0}O^\top\right)^{\otimes t}\right]\\
			&= \mathbb{E}\left[\sum_{x,y\in[d]^t}\prod_{k=1}^{t}\frac{W_{x_k}}{\|W\|}\frac{W_{y_k}}{\|W\|}\ketbra{x}{y}\right]\,.
		\end{align}
		We then have by linearity and by definition of \(n_i\)
		\begin{equation}
			\rrhodt = \sum_{x,y\in[d]^t}\mathbb{E}\left[\prod_{i=0}^{d-1}\frac{W_i^{n_i(x, y)}}{\|W\|^{n_i(x, y)}}\right]\ketbra{x}{y}\,.
		\end{equation}
            We then have the desired result via \Cref{lem:expectationY}.
	\end{proof}
	\subsection{Proof of Theorem~\ref{thm:eigenvalues}}
	\label{app:proof_eigenvalues}
        Now that we have an expression for the coefficients of \(\rrhodt\) and that we know an expression of an eigenvector \(\ket{\psi_\lambda}\) for each of its unique eigenvalues, we can solve for \(\rrhodt\ket{\psi_\lambda}=\lambda\ket{\psi_\lambda}\) to find said eigenvalues. This yields the following theorem.
	\eigenvalues*
	\begin{proof}
		\label{proof:eigenvalues-real-haar-random}
		For the sake of conciseness, let us define
		\begin{equation}
			\normrrho\defed\sum_{x\sim y}\prod_{i=1}^d\left[n_i(x,y)-1\right]!!\ketbra{x}{y}\,,
		\end{equation}
            and so $\rrhodt = \frac{(d-2)!!}{(d+2t-2)!!}\,\normrrho$
		Let \(k_0\in\left[1+\left\lfloor\frac{t}{2}\right\rfloor\right]\). Since we know the multiplicity of the eigenvalue associated to \(\mathcal{H}_d^{t-2k_0}\) via \Cref{eq:dim-subspaces}, the only thing we have to do is to compute this eigenvalue. Let us take \(\alpha\) so that we have
		\begin{equation}
                \label{eq:belongstoht-2k0}
			\sum_{\ell_0+\cdots+\ell_{d-1}=t-2k_0}\alpha_{\ell_0,\cdots,\ell_{d-1}}X_0^{\ell_0}\cdots X_{d-1}^{\ell_{d-1}}\in\mathcal{H}_d^{t-2k_0}
		\end{equation}
		We now want to multiply this polynomial by \(q^{k_0}\) to construct our eigenvector, which gives us
		\begin{subequations}
			\begin{align}
				&\left(\sum_{j=0}^{d-1}X_j^2\right)^{k_0}\sum_{\ell_0+\cdots+\ell_{d-1}=t-2k_0}\alpha_{\ell_0,\cdots,\ell_{d-1}}X_0^{\ell_0}\cdots X_{d-1}^{\ell_{d-1}}\\
				={}&\sum_{k_{d-1}\leqslant\cdots\leqslant k_1=0}^{k_0}\left(\prod_{j=0}^{d-2}\binom{k_j}{k_{j+1}}X_j^{2k_j-2k_{j+1}}\right)X_{d-1}^{2k_{d-1}}\sum_{\ell_0+\cdots+\ell_{d-1}=t-2k_0}\alpha_{\ell_0,\cdots,\ell_{d-1}}X_0^{\ell_0}\cdots X_{d-1}^{\ell_{d-1}}\\
				={}&\sum_{k_{d-1}\leqslant\cdots\leqslant k_1=0}^{k_0}\left(\prod_{j=0}^{d-2}\binom{k_j}{k_{j+1}}\right)\sum_{\ell_0+\cdots+\ell_{d-1}=t-2k_0}\alpha_{\ell_0,\cdots,\ell_{d-1}}X_{d-1}^{\ell_{d-1}+2k_{d-1}}\prod_{j=0}^{d-2}X_j^{\ell_j+2k_j-2k_{j+1}}\,.
			\end{align}
		\end{subequations}
		Now, for indexes \(\ell_0,\cdots,\ell_{d-1}\) that sum up to \(t\), let us denote the (not normalized) state
		\begin{subequations}
			\begin{align}
				\ket{\varphi_{\ell_0,\cdots,\ell_{d-1}}}&=\sum_{\sigma\in\mathfrak{S}_t}\ket{\sigma\left(0^{\ell_0},\cdots,(d-1)^{\ell_{d-1}}\right)}\\
				\label{eq:phimult}
				&= \left(\prod_{j=0}^{d-1}\ell_j!\right)\sum_{z\in\mathcal{C}\left(0^{\ell_0},\cdots,(d-1)^{\ell_{d-2}}\right)}\ket{z}
			\end{align}
		\end{subequations}
		with \(\mathcal{C}\left(0^\ell_0,\cdots,(d-1)^{\ell_{d-1}}\right)\) being the equivalence class of \(\left(0^{\ell_0},\cdots,(d-1)^{\ell_{d-1}}\right)\), with the equivalence relation testing whether two strings are equal up to a permutation. It will be useful to mention that
		\begin{equation}
			\label{eq:sizeequicperm}
			\left|(\mathcal{C}\left(0^{\ell_0},\cdots,(d-1)^{\ell_{d-1}}\right)\right|=\frac{t!}{\prod\limits_{j=0}^{d-1}\ell_i!}\,.
		\end{equation}
		Thus, the following vector \(\ket{\lambda}\) is an eigenvector of \(\hat{\rho}_{\mathcal{R}_d^t}\)
		\begin{equation}
			\sum_{k_{d-1}\leqslant\cdots\leqslant k_1=0}^{k_0}\left(\prod_{j=0}^{d-2}\binom{k_j}{k_{j+1}}\right)\sum_{\ell_0+\cdots+\ell_{d-1}=t-2k_0}\alpha_{\ell_0,\cdots,\ell_{d-1}}\ket{\varphi_{\ell_0+2k_0-2k_1,\cdots,\ell_{d-2}+2k_{d-2}-2k_{d-1},\ell_{d-1}+2k_{d-1}}}\,.
		\end{equation}
		To get the associated eigenvalue, it's thus enough to multiply \(\hat{\rho}_{\mathcal{R}_d^t}\) by this vector, which gives us
		\begin{subequations}
			\begin{equation}
                    \begin{split}
    				\hat{\rho}_{\mathcal{R}_d^t}\ket{\lambda}=\Biggl[\sum_{x\sim y}&\prod_{\ell=0}^{d-1}\left[n_\ell(x,y)-1\right]!!\ketbra{x}{y}\sum_{k_{d-1}\leqslant\cdots\leqslant k_1=0}^{k_0}\left(\prod_{j=0}^{d-2}\binom{k_j}{k_{j+1}}\right)\\&\sum_{\ell_0+\cdots+\ell_{d-1}=t-2k_0}\alpha_{\ell_0,\cdots,\ell_{d-1}}\ket{\varphi_{\ell_0+2k_0-2k_1,\cdots,\ell_{d-2}+2k_{d-2}-2k_{d-1},\ell_{d-1}+2k_{d-1}}}\Biggl]
                    \end{split}
			\end{equation}
			The inner product between \(\ket{y}\) and \(\ket{\varphi_{\ell_0+2k_0-2k_1,\cdots,\ell_{d-2}+2k_{d-2}-2k_{d-1},\ell_{d-1}+2k_{d-1}}}\) keeps only the vectors $\ket{y}$ that are in the right equivalence class. Summing over this equivalence class and multiplying by the scalar product then gives a \(t!\) factor, as per \Cref{eq:phimult,eq:sizeequicperm}.
			
			We also have to consider the \(x\) that are in relation with such \(y\). In order to do so, the number of \(i\) in such a \(x\) must have the same parity as that of \(\ell_i\). Let \(b_i=\ell_i\mod2\) be the parity of \(\ell_i\) for a given \(i\). We can describe such an \(x\) with a tuple \(m_0,\cdots,m_{d-1}\), with \(2m_i+b_i\) being the number of \(i\) in \(x\). Summing over all the possible \(x\), we see that we sum over the same basis states as in \(\ket{\varphi_{2m_0+b_0,\cdots,2m_{d-1}+b_{d-1}}}\). All in all, this gives
			\begin{equation}
                    \begin{split}
                        \hat{\rho}_{\mathcal{R}_d^t}\ket{\lambda}=\Biggl[&t!\sum_{\ell_0+\cdots+\ell_{d-1}=t-2k_0}\alpha_{\ell_0,\cdots,\ell_{d-1}}\sum_{\substack{k_{d-1}\leqslant\cdots\leqslant k_1=0\\m_0+\cdots+m_{d-1}=\frac{t-p(\ell)}{2}\\m_i\geqslant0}}^{k_0}\prod_{j=0}^{d-2}\binom{k_j}{k_{j+1}}\\&\frac{\left(\ell_{d-1}+2k_{d-1}+2m_{d-1}+b_{d-1}-1\right)!!\prod\limits_{j=0}^{d-2}\left(\ell_j+2k_j-2k_{j+1}+2m_j+b_j-1\right)!!}{\prod\limits_{j=0}^{d-1}\left(2m_j+b_j\right)!}\\&\ket{\varphi_{2m_0+b_0,\cdots,2m_{d-1}+b_{d-1}}}\Biggr]
                    \end{split}
                \end{equation}
		\end{subequations}
		with \(p(\ell)=\sum_ib_i\) being the number of odd elements in \(\ell\). Now, since \(\ket{\lambda}\) is an eigenvector of \(\hat{\rho}_{\mathcal{R}_d^t}\), we have for all \(\ket{\psi}\) that isn't orthogonal with \(\ket{\lambda}\) that
		\begin{equation}
			\lambda=\frac{\braopket{\psi}{\hat{\rho}_{\mathcal{R}_d^t}}{\lambda}}{\inner{\psi}{\lambda}}\,.
		\end{equation}
		In particular, we can choose
		\begin{equation}
			\ket{\psi}=\sum_{\substack{b\in\{0, 1\}^{d-1}\\h(b)\leqslant t-2k_0}}\ket{\varphi_{b_0,\cdots,b_{d-2},t-h(b)}}
		\end{equation}
		with \(h\) being the Hamming weight. This forces \(m_{d-1}=\frac{t-p(\ell)}{2}\) and \(m_i=0\) for \(i<d-1\) in the above summation. Note that we then have
		\begin{equation}
			\inner{\psi}{\lambda}=\sum_{\substack{b\in\{0, 1\}^{d-1}\\h(b)\leqslant t-2k_0}}\alpha_{b_0,\cdots,b_{d-2},t-2k_0-h(b)}t!(t-h(b))!
		\end{equation}
		All in all, this gives us:
		\begin{subequations}
			\begin{align}
                    \begin{split}
				\lambda = \Biggl[&\frac{(t!)^2}{\inner{\psi}{\lambda}}\sum_{\ell_0+\cdots+\ell_{d-1}=t-2k_0}\alpha_{\ell_0,\cdots,\ell_{d-1}}\sum_{k_{d-1}\leqslant\cdots\leqslant k_1=0}^{k_0}\prod_{j=0}^{d-2}\binom{k_j}{k_{j+1}}\\&\frac{\left(\ell_{d-1}+2k_{d-1}+t-p(\ell)+b_{d-1}-1\right)!!\prod\limits_{j=0}^{d-2}\left(\ell_j+2k_j-2k_{j+1}+b_j-1\right)!!}{\left(t-p(\ell)+b_{d-1}\right)!}\left(t-p(\ell)+b_{d-1}\right)!\Biggr]
                \end{split}\\
                \begin{split}
				=\Biggl[&\frac{(t!)^2}{\inner{\psi}{\lambda}}\sum_{\ell_0+\cdots+\ell_{d-1}=t-2k_0}\alpha_{\ell_0,\cdots,\ell_{d-1}}\sum_{k_{d-1}\leqslant\cdots\leqslant k_1=0}^{k_0}\left(\prod_{j=0}^{d-2}\binom{k_j}{k_{j+1}}\left(\ell_j+2k_j-2k_{j+1}+b_j-1\right)!!\right)\\&\left(\ell_{d-1}+2k_{d-1}+t-p(\ell)+b_{d-1}-1\right)!!\Biggr]
                \end{split}
			\end{align}
		\end{subequations}
		Note that this relation holds for any \(\alpha\) such that \Cref{eq:belongstoht-2k0} holds.
		
		Let us assume \(d\geqslant3\) for now. For \(i\in \left[\left\lfloor\frac{t}{2}\right\rfloor-k_0+1\right]\), let us take \(\alpha_{0,\cdots,0,2i,t-2k_0-2i} = (-1)^i\binom{t-2k_0}{2i}\) and nil on all other indexes. This satisfies \Cref{eq:belongstoht-2k0}. We have in this case \(\inner{\psi}{\lambda}=(t!)^2\). After successive applications of \Cref{lem:doublefactorialbinomialformula} and \Cref{cor:multinome_double_factorial}, the expression of \(\lambda\) then becomes
		\begin{equation}
			\label{eq:expression_with_unknown_sum}
			\lambda = \frac{\left(2t+d-2\right)!!}{\left(2t-2k_0+d-2\right)!!}\sum_{i=0}^{\left\lfloor\frac{t}{2}\right\rfloor-k_0}(-1)^{i}\binom{t-2k_0}{2i}(2i-1)!!\left(2t-2k_0-2i-1\right)!!\,.
		\end{equation}
		For \(d=2\), the computations read:
		\begin{subequations}
			\begin{align}
				\lambda &= \sum_{i=0}^{\left\lfloor\frac{t}{2}\right\rfloor-k_0}(-1)^{i}\binom{t-2k_0}{2i}\sum_{k_1=0}^{k_0}\binom{k_0}{k_1}\left(2i+2k_0-2k_1-1\right)!!\left(2t-2k_0-2i+2k_1-1\right)!!\\
				&=\sum_{i=0}^{\left\lfloor\frac{t}{2}\right\rfloor-k_0}(-1)^{i}\binom{t-2k_0}{2i}\frac{\left(2t-2k_0-2i-1\right)!!\left(2i-1\right)!!}{\left(2t-2k_0\right)!!}\left(2t\right)!!
			\end{align}
		\end{subequations}
		via \Cref{lem:doublefactorialbinomialformula}, which agrees with \Cref{eq:expression_with_unknown_sum} for \(d=2\). Thus, we now only have to evaluate the rightmost sum. Note that this sum is independent of \(d\). Thus, if we can compute \(\lambda\) for \(t\) and \(k_0\) being arbitrary large, then we would have computed the value of this sum, which would conclude the proof.
		
		To do so, let us take \(d=t-2k_0\) and take \(\alpha_{1,\cdots,1}=1\) and be nil otherwise. It once again satisfies \Cref{eq:belongstoht-2k0}, and we have in that case \(\inner{\psi}{\lambda}=t!\left(2k_0+1\right)!\). We thus have
		\begin{subequations}
			\begin{align}
				\lambda &=\frac{t!}{\left(2k_0+1\right)!}\sum_{k_{d-1}\leqslant\cdots\leqslant k_{1}=0}^{k_{0}}\left(\prod_{j=0}^{d-2}\binom{k_j}{k_{j+1}}\left(2k_j-2k_{j+1}+1\right)!!\right)\left(2k_{d-1}+2k_0+1\right)!!\\
				&=\frac{t!\left(4k_0+1+3(d-1)\right)!!\left(2k_0+1\right)!!}{\left(2k_0+1\right)!\left(2k_0+1+3(d-1)\right)!!}\\
				&= \frac{t!\left(4k_0+1+3\left(t-2k_0-1\right)\right)!!}{\left(2k_0\right)!!\left(2k_0+1+3\left(t-2k_0-1\right)\right)!!}\\
				&= \frac{t!\left(3t-2k_0-2\right)!!}{\left(2k_0\right)!!\left(3t-4k_0-2\right)!!}
			\end{align}
		\end{subequations}
		where we used \Cref{cor:multinome_double_factorial} on the second line. But now, we know from \Cref{eq:expression_with_unknown_sum} that we have
		\begin{subequations}
			\begin{align}
                \begin{split}
    				&\frac{\left(3t-2k_0-2\right)!!}{\left(3t-4k_0-2\right)!!}\sum_{i=0}^{\left\lfloor\frac{t}{2}\right\rfloor-k_0}(-1)^{i}\binom{t-2k_0}{2i}(2i-1)!!\left(2t-2k_0-2i-1\right)!!\\
                    ={}&\frac{t!\left(3t-2k_0-2\right)!!}{\left(2k_0\right)!!\left(3t-4k_0-2\right)!!}
                \end{split}\\
				\iff&\sum_{i=0}^{\left\lfloor\frac{t}{2}\right\rfloor-k_0}(-1)^{i}\binom{t-2k_0}{2i}(2i-1)!!\left(2t-2k_0-2i-1\right)!!=\frac{t!}{\left(2k_0\right)!!}\,.
			\end{align}
		\end{subequations}
		This finally allows us to conclude that
		\begin{equation}
			\lambda = \frac{t!\left(2t+d-2\right)!!}{\left(2k_0\right)!!\left(2t-2k_0+d-2\right)!!}\,.
		\end{equation}
	\end{proof}
    \subsection{Generating the eigenvectors of \texorpdfstring{\(\rrhodt\)}{the density matrix resulting from sampling t copies of a d-dimensional real-valued quantum state according to the Haar measure}}
    \label{app:example_spetrctal_decomposition}
        In this section, we will describe a method to generate the eigenvectors of \(\rrhodt\) for arbitrary \(d\) and \(t\). First of all, let us recall that \(\rrhodt\) is an operator on \(\symsubspacedt{R}\), and that \(\symsubspacedt{R}\) is isomorphic to the space \(\homogeneousdt\) of homogeneous polynomials of \(d\) variables of degree \(t\), with the isomorphism being described in~\Cref{eq:isomorphism}. We also recall that \(\homogeneousdt\) can be decomposed into
        \begin{equation}
            \homogeneousdt=\harmonicd{t}\oplus q\harmonicd{t-2}\oplus q^2\harmonicd{t-4}\oplus\cdots
        \end{equation}
        where we've denoted
        \begin{equation}
            q\defed\sum_{i=0}^{d-1}X_i^2
        \end{equation}
        and with \(\harmonicd{t}\) being the space of harmonic homogeneous polynomials of \(d\) variables and of degree \(t\). \Cref{thm:eigenvalues} states that \(\rrhodt\) is proportional to the identity on each of these subspaces. In order to generate its eigenvectors, we thus have to generate an orthonormal basis of \(q^k\harmonicd{t-2k}\) and then apply the isomorphism to translate these polynomials into vectors in \(\symsubspacedt{R}\). All in all, in order to generate the eigenvectors of \(\rrhodt\), one can follow~\Cref{alg:eigenvectors}.
        \begin{algorithm}[ht]
        	\caption{Generation of the eigenvectors of \(\rrhodt\)}\label{alg:eigenvectors}
        	\begin{algorithmic}
        		\Require $d\geqslant2$, $t\geqslant1$
        		\State $V \gets \varnothing$
        		\For{$k\in\left[1+\left\lfloor\frac{t}{2}\right\rfloor\right]$}
                    \State $B\gets \texttt{HarmonicBasis[}d\texttt{,}t-2k\texttt{]}$
                    \For{$P\in B$}
                        \State $P\gets \sqrt{\frac{t!(2t-4k+d-2)!!}{(t-2k)!(2k)!!(2t-2k+d-2)!!}}\,q^kP$\Comment{We have to renormalize \(q^kP\)}
                    \EndFor
                    \State $V \gets V\cup B$
                \EndFor
                \State \textbf{return} $\texttt{Isomorphism[}\homogeneousdt\to\symsubspacedt{R}\texttt{,}V\texttt{]}$
        	\end{algorithmic}
        \end{algorithm}
        
        In order to run~\Cref{alg:eigenvectors}, the only unspecified part is how to generate an orthonormal basis of \(\harmonicd{t}\) for arbitrary \(d\) and \(t\). For \(d=2\), such a basis has been given in~\Cref{eq:orthonormal_basis_h2t}. We thus move on to describing an algorithm to generate such a basis in the case \(d\geqslant3\).

        \subsubsection{Presentation of the algorithm}
        First of all, we will need a method to generate a basis of \(\harmonicd{t}\). To that end, we will use the following theorem.
        \begin{theorem}[Adapted from~{\cite[Theorem~5.25]{ABW13}}]
            Let \(d>2\) and \(t\) be two natural numbers. The set
            \begin{equation}
                \left\{K\left[\frac{\partial^{\alpha_0}}{\partial X_0^{\alpha_0}}\cdots\frac{\partial^{\alpha_{d-1}}}{\partial X_{d-1}^{\alpha_{d-1}}}q^{1-\frac{d}{2}}\right]\middle|\alpha_{d-1}\in\{0, 1\}\land \alpha\in\mathbb{N}^d\land \sum_{k=0}^{d-1}\alpha_k=t\right\}
            \end{equation}
            is a basis of \(\harmonicd{t}\), with \(K\) being the Kelvin transform, which is defined as
            \begin{equation}
                K[f](X) \defed q^{1-\frac{d}{2}}f\left(\frac{X_0}{q},\cdots,\frac{X_{d-1}}{q}\right)
            \end{equation}
            and where we've reused the notation
            \begin{equation}
                q\defed\sum_{k=0}^{d-1}X_k^2\,.
            \end{equation}
        \end{theorem}

        Using the Faà di Bruno formula~\cite[Proposition~1]{Har06}, differentiating \(q^{1-\frac{d}{2}}\) with respect to the multi-index \(\alpha\) yields
        \begin{equation}
            \frac{\partial^{\alpha_0}}{\partial X_0^{\alpha_0}}\cdots\frac{\partial^{\alpha_{d-1}}}{\partial X_{d-1}^{\alpha_{d-1}}}q^{1-\frac{d}{2}} = \sum_{\pi\in\Pi}\left(-\frac12\right)^{|\pi|}\frac{(d+2|\pi|-4)!!}{(d-4)!!}q^{-\frac{d}{2}-|\pi|+1}\prod_{B\in\pi}\frac{\partial^{|B|}}{\prod\limits_{j\in B}\partial Y_j}q
        \end{equation}
        where \(\Pi\) is the set of partitions of the set \(\left\{1,\cdots, t\right\}\) and where we've defined \(Y_1=\cdots=Y_{\alpha_0}=X_0\), \(Y_{\alpha_0+1}=\cdots=Y_{\alpha_0+\alpha_1}=X_1\), \(\dots\) Recall that a partition of \(\left\{1,\cdots,t\right\}\) is a set of disjoints, non-empty sets such that their union is \(\left\{1,\cdots,t\right\}\). We want to exhibit the non-zero terms in this expression. Note that any third derivative of \(q\) is nil, and that the only second derivatives of \(q\) that aren't nil are those that differentiate with respect to the same variable twice. The only remaining partitions are thus those that contain no set of more than 2 elements, and each set with 2 elements has its elements belonging to the same subset of \(\left\{1,\cdots,t\right\}\) that corresponds to a given \(\alpha_i\). For instance, if \(\alpha_0=4\), then partitions containing \(\left\{1, 4\right\}\) or \(\left\{1, 2\right\}\) for instance would yield a non-zero term, but \(\left\{1, 5\right\}\) will not, since that would correspond to differentiate \(q\) with respect to \(X_0X_1\).

        One way to see the remaining partitions is to see it as a union over the partitions of the sets corresponding to the \(\alpha_i\), with these partitions having no terms with more than 2 elements. That is, we have
        \begin{equation}
            \frac{\partial^{\alpha_0}}{\partial X_0^{\alpha_0}}\cdots\frac{\partial^{\alpha_{d-1}}}{\partial X_{d-1}^{\alpha_{d-1}}}q^{1-\frac{d}{2}} = \sum_{\pi_0\in\Pi_0}\cdots\sum_{\pi_{d-1}\in\Pi_{d-1}}(-1)^{|\pi|}\frac{\left(d+2|\pi|-4\right)!!}{(d-4)!!}q^{-\frac{d}{2}-|\pi|+1}\prod_{i=0}^{d-1}X_i^{\alpha_i-2k_i}
        \end{equation}
        with \(\Pi_i\) being the set of partitions of \(\left\{1,\cdots,\alpha_i\right\}\), \(\pi=\bigcup\limits_{i=0}^{d-1}\pi_i\) and \(k_i\) being the number of sets of size \(2\) in \(\pi_i\) (that is, \(k_i=\alpha_i-\left|\pi_i\right|\)). Applying the Kelvin transform to this term then yields
        \begin{equation}
            K\left[\frac{\partial^{\alpha_0}}{\partial X_0^{\alpha_0}}\cdots\frac{\partial^{\alpha_{d-1}}}{\partial X_{d-1}^{\alpha_{d-1}}}q^{1-\frac{d}{2}}\right]=\sum_{\pi_0\in\Pi_0}\cdots\sum_{\pi_{d-1}\in\Pi_{d-1}}(-1)^{|\pi|}\frac{\left(d+2|\pi|-4\right)!!}{(d-4)!!}q^{t-|\pi|}\prod_{i=0}^{d-1}X_i^{\alpha_i-2k_i}\,.
        \end{equation}
        Since \(\alpha_i-2k_i=2\left|\pi_i\right|-\alpha_i\), note that the summand only depends on the size of the partitions. Note that for a fixed \(k_i\), the number of partitions of \(\left\{1,\cdots,\alpha_i\right\}\) having \(k_i\) elements of size \(2\) and no element of size larger than \(2\) is \(\frac{\alpha_i!}{\left(\alpha_i-2k_i\right)!\left(2k_i\right)!!}\). As such, we can regroup the terms with identical monomials to get
        \begin{equation}
            \label{eq:formula_basis_harmonic_d_t}
            \begin{split}
                &K\left[\frac{\partial^{\alpha_0}}{\partial X_0^{\alpha_0}}\cdots\frac{\partial^{\alpha_{d-1}}}{\partial X_{d-1}^{\alpha_{d-1}}}q^{1-\frac{d}{2}}\right]\\
                ={}&\sum_{k_0=0}^{\left\lfloor\frac{\alpha_0}{2}\right\rfloor}\cdots\sum_{k_{d-1}=0}^{\left\lfloor\frac{\alpha_{d-1}}{2}\right\rfloor}(-1)^{t-\sum\limits_{i=0}^{d-1}k_i}\frac{\left(d+2t-2\sum\limits_{i=0}^{d-1}k_i-4\right)!!}{(d-4)!!}q^{\sum\limits_{i=0}^{d-1}k_i}\prod_{i=0}^{d-1}\frac{\alpha_i!}{\left(\alpha_i-2k_i\right)!\left(2k_i\right)!!}X_i^{\alpha_i-2k_i}\,.
            \end{split}
        \end{equation}
        We are now able to write an algorithm to generate an orthonormal basis of \(\harmonicd{t}\). This algorithm is presented on~\Cref{alg:harmonic_basis}.
        \begin{algorithm}[ht]
        	\caption{Generation of an orthonormal basis of \(\harmonicd{t}\)}\label{alg:harmonic_basis}
        	\begin{algorithmic}
        		\Require $d\geqslant3$, $t\geqslant1$
        		\State $B \gets \varnothing$
        		\For{$\alpha\in\left\{\beta\in\mathbb{N}^d\middle|\beta_{d-1}\in\{0, 1\}\land\sum\limits_{k=0}^{d-1}\beta_k=t\right\}$}\Comment{We iterate over the valid multi-indices}
                    \State $P\gets 0$
                    \For{$\left(k_0,\cdots,k_{d-2}\right)\in\left[\left\lfloor\frac{\alpha_0}{2}\right\rfloor+1\right]\times\cdots\times\left[\left\lfloor\frac{\alpha_{d-2}}{2}\right\rfloor+1\right]$}
                        \State $k_{d-1}\gets0$\Comment{Since $\alpha_{d-1}\in\{0, 1\}$, we necessarily have \(k_{d-1}=0\)}
                        \State $K\gets\sum\limits_{i=0}^{d-1}k_i$
                        \State $P\gets P + (-1)^{t-K}\frac{\left(d+2t-2K-4\right)!!}{(d-4)!!}q^{K}\prod\limits_{i=0}^{d-1}\frac{\alpha_i!}{\left(\alpha_i-2k_i\right)!\left(2k_i\right)!!}X_i^{\alpha_i-2k_i}$
                    \EndFor
                    \State $B \gets B\cup \{P\}$
                \EndFor
                \State \textbf{return} $\text{Orthonormalize}[B]$
        	\end{algorithmic}
        \end{algorithm}

        Let us now study the complexity of this algorithm. Computing \(K\) is done in \(\bigo{d}\) time. The double factorials can be computed in \(\bigo{1}\) time via their approximation using the \(\Gamma\) function. The product is computed in \(\bigo{d}\) time and \(q^K\) can be computed in at most \(d^2\log_2(t)\) time using fast exponentiation. As such, the entire body of the inner loop runs in \(\bigo{d^2\ln(t)}\) time. For a given multi-index \(\alpha\), this means that the inner loop runs in
        \begin{equation}
            \bigo{d^2\ln(t)\prod_{i=0}^{d-1}\left(\left\lfloor\frac{\alpha_i}{2}\right\rfloor+1\right)}
        \end{equation}
        time. We have
        \begin{equation}
            \prod_{i=0}^{d-1}\left(\left\lfloor\frac{\alpha_i}{2}\right\rfloor+1\right)\leqslant\prod_{i=0}^{d-1}\left(\frac{\alpha_i}{2}+1\right)\leqslant\left(\frac{\frac{t}{2}+d-1}{d-1}\right)^{d-1}=\mathrm{e}^{(d-1)\ln\left(1+\frac{t}{2(d-1)}\right)}\leqslant\mathrm{e}^{\frac{t}{2}}
        \end{equation}
        where the first inequality comes from \(\lfloor x\rfloor\leqslant x\), the second is obtained via the AM-GM inequality, and the last one uses \(\ln(1+x)\leqslant x\). As such, the inner loop runs in at most \(\bigo{d^2\ln(t)\mathrm{e}^{\frac{t}{2}}}\) time. Summing over the valid multi-indices, we get that generating the basis is done in
        \begin{equation}
            \bigo{\dim\left(\harmonicd{t}\right)d^2\ln(t)\mathrm{e}^{\frac{t}{2}}}
        \end{equation}
        time.

        The Gram-Schmidt orthonormalization is also costly: the inner product takes linear time in the sizes of the polynomials (which we've upper-bounded by \(\mathrm{e}^{\frac{t}{2}}\)) and linear time in the sizes of the monomials, and \(\bigo{\dim\left(\harmonicd{t}\right)^2}\) operations must be performed in total. Thus, the Gram-Schmidt process runs in at most
        \begin{equation}
            \bigo{\dim\left(\harmonicd{t}\right)^2d\mathrm{e}^{\frac{t}{2}}}
        \end{equation}
        time, which supersedes the previous loop and thus gives the final complexity of the algorithm.

        We note that this method for generating orthonormal bases of \(\harmonicd{t}\) is implemented in the \texttt{HFT.m} Mathematica package in the \texttt{basisH} function\footnote{It's not clear that the implementation in this package will have the same complexity as the one presented here, since our method merges the differentiation and the Kelvin transform steps, instead of doing them sequentially.}~\cite{HFT}.

        \subsubsection{First example: computing the spectral decomposition of \texorpdfstring{\(\rrho{3}{3}\)}{the density matrix associated with sampling 3 copies of a 3-dimensional quantum state according to the Orthogonal Haar measure} using our algorithm}
        Let us see how to compute the eigenvectors of \(\rrho{3}{3}\) using this method. We start with the first eigenspace, which is \(\harmonic{3}{3}\). We thus want to use~\Cref{alg:harmonic_basis} in order to generate an orthonormal basis of \(\harmonic{3}{3}\).

        We start with \(\alpha=(3, 0, 0)\). The inner loop will go through \((0, 0)\) and then \((1, 0)\), which gives us the following polynomial:
        \begin{equation}
            P_{(3, 0, 0)} = -15\,X_0^3+9\,X_0q = -6\,X_0^3+9\,X_0^2X_1+9\,X_0^2X_2\,.
        \end{equation}

        To illustrate, if we now pick \(\alpha=(1, 1, 1)\), then the inner loop performs a single iteration at \((0, 0)\) and gives
        \begin{equation}
            P_{(1, 1, 1)} = -15\,X_0X_1X_2\,.
        \end{equation}
        
        By iterating this process over the 7 possible multi-indices \(\alpha\), we get the following basis of \(\harmonic{3}{3}\)
        \begin{equation}
            \begin{Bmatrix}
                -6\,X_0^3+9\,X_0X_1^2+9\,X_0X_2^2,\\
                9\,X_0^2X_1-6\,X_1^3+9\,X_1X_2^2,\\
                 -12\,X_0^2X_1+3\,X_1^3+3\,X_1X_2^2,\\
                 -12\,X_0^2X_2+3\,X_1^2X_2+3\,X_2^3,\\
                 3\,X_0^3-12\,X_0X_1^2+3\,X_0X_2^2,\\
                 -12\,X_1^2X_2+3\,X_0^2X_2+3\,X_2^3,\\
                 -15\,X_0X_1X_2
            \end{Bmatrix}\,.
        \end{equation}
        We can then orthonormalize this basis with respect to the Bombieri inner product using the Gram-Schmidt process, which gives us the following orthonormal basis
        \begin{equation}
            \begin{Bmatrix}
                \frac{2\,X_0^3-3\,X_0X_1^2-3\,X_0X_2^2}{\sqrt{10}},\\
                \frac{3\,X_0^2X_1-2\,X_1^3+3\,X_1X_2^2}{\sqrt{10}},\\
                 \frac{3\,X_0^2X_1-3\,X_1X_2^2}{\sqrt{6}},\\
                 \frac{12\,X_0^2X_2-3\,X_1^2X_2-3\,X_2^3}{\sqrt{60}},\\
                 \frac{3\,X_0X_1^2-3\,X_0X_2^2}{\sqrt{6}}\\
                 \frac{3\,X_1^2X_2-X_2^3}{2},\\
                 \sqrt{6}\,X_0X_1X_2
            \end{Bmatrix}\,.
        \end{equation}
        Applying the isomorphism described in~\Cref{eq:isomorphism}, we have that the \(7\)-dimensional eigenspace of \(\rrho{3}{3}\) associated to the eigenvalue \(\frac{2}{35}\) is generated by
        \begin{equation}
            \begin{Bmatrix}
                \frac{2\,\ket{000}-\ket{011}-\ket{101}-\ket{110}-\ket{022}-\ket{202}-\ket{220}}{\sqrt{10}},\\
                \frac{\ket{001}+\ket{010}+\ket{100}-2\,\ket{111}+\ket{122}+\ket{212}+\ket{221}}{\sqrt{10}},\\
                 \frac{\ket{001}+\ket{010}+\ket{100}-\ket{122}-\ket{212}-\ket{221}}{\sqrt{6}},\\
                 \frac{4\,\ket{002}+4\,\ket{020}+4\,\ket{200}-\ket{112}-\ket{121}-\ket{211}-3\,\ket{222}}{\sqrt{60}},\\
                 \frac{\ket{011}+\ket{101}+\ket{110}-\ket{022}-\ket{202}-\ket{220}}{\sqrt{6}},\\
                 \frac{\ket{112}+\ket{121}+\ket{211}-\ket{222}}{2},\\
                 \frac{\ket{012}+\ket{021}+\ket{102}+\ket{120}+\ket{201}+\ket{210}}{\sqrt{6}}
            \end{Bmatrix}\,.
        \end{equation}
        Similarly, a basis of \(q\harmonic{3}{1}\) is simply \(\left\{qX_0,qX_1,qX_2\right\}\). By renormalizing it and applying the isomorphism to it, we get that the \(3\)-dimensional eigenspace of \(\rrho{3}{3}\) associated to the eigenvalue \(\frac{1}{5}\) is generated by
        \begin{equation}
            \begin{Bmatrix}
                \frac{3\,\ket{000}+\ket{011}+\ket{101}+\ket{110}+\ket{022}+\ket{202}+\ket{220}}{\sqrt{15}},\\
            \frac{3\,\ket{111}+\ket{100}+\ket{010}+\ket{001}+\ket{122}+\ket{212}+\ket{221}}{\sqrt{15}},\\
            \frac{3\,\ket{222}+\ket{200}+\ket{020}+\ket{002}+\ket{211}+\ket{212}+\ket{112}}{\sqrt{15}}
            \end{Bmatrix}\,.
        \end{equation}

        \subsubsection{Second example: computing the spectral decomposition of \texorpdfstring{\(\rrho{3}{2}\)}{the density matrix associated with sampling 3 copies of a 3-dimensional quantum state according to the Orthogonal Haar measure} using an orthogonal basis of spherical harmonics}
        Another technique we could have used is starting from an orthogonal\footnote{Note that there are several conventions for the inner product used to define orthogonality for spherical harmonics, but they are all proportional to the Bombieri inner product. This preserves orthogonality, but one may have to renormalize.} basis of spherical harmonics. For instance, such a basis for \(d=3\) and \(t=2\) is given by
        \begin{equation}
            \begin{Bmatrix}
                \sin^2(\theta)\mathrm{e}^{-2\mathrm{i}\varphi},\\
                \sin(\theta)\cos(\theta)\mathrm{e}^{-\mathrm{i}\varphi},\\
                3\,\cos^2(\theta)-1,\\
                \sin(\theta)\cos(\theta)\mathrm{e}^{\mathrm{i}\varphi},\\
                \sin^2(\theta)\mathrm{e}^{2\mathrm{i}\varphi}
            \end{Bmatrix}\,.
        \end{equation}
        Spherical harmonics are usually expressed in terms of spherical coordinates. We thus want to translate them to cartesian coordinates. We thus take
        \begin{equation}
            \begin{array}{ll}
            \cos(\varphi) = \frac{X_0}{\sqrt{X_0^2+X_1^2}}&\sin(\varphi) = \frac{X_1}{\sqrt{X_0^2+X_1^2}}\\
            \cos(\theta) = \frac{X_2}{\sqrt{X_0^2+X_1^2+X_2^2}}&\sin(\theta)=\sqrt{\frac{X_0^2+X_1^2}{X_0^2+X_1^2+X_2^2}}
            \end{array}
        \end{equation}
        which gives us
        \begin{equation}
            \begin{Bmatrix}
                \frac{X_0^2-X_1^2-2\mathrm{i}\,X_0X_1}{X_0^2+X_1^2+X_2^2},\\
                \frac{X_2\left(X_0-\mathrm{i}\,X_1\right)}{X_0^2+X_1^2+X_2^2},\\
                \frac{2\,X_2^2-X_0^2-X_1^2}{X_0^2+X_1^2+X_2^2},\\
                \frac{X_2\left(X_0+\mathrm{i}\,X_1\right)}{X_0^2+X_1^2+X_2^2},\\
                \frac{X_0^2-X_1^2+2\mathrm{i}\,X_0X_1}{X_0^2+X_1^2+X_2^2}
            \end{Bmatrix}\,.
        \end{equation}
        Now, if \(\left\{Y_i\right\}_i\) is an orthogonal basis of spherical harmonics, then \(\left\{q^{\frac{t}{2}}Y_i\left(\frac{X}{\sqrt{q}}\right)\right\}_i\) is an orthogonal basis of \(\harmonicd{t}\)~\cite{DX13}. This gives us
        \begin{equation}
            \begin{Bmatrix}
                X_0^2-X_1^2-2\mathrm{i}\,X_0X_1,\\
                X_2\left(X_0-\mathrm{i}\,X_1\right),\\
                2\,X_2^2-X_0^2-X_1^2,\\
                X_2\left(X_0+\mathrm{i}\,X_1\right),\\
                X_0^2-X_1^2+2\mathrm{i}\,X_0X_1
            \end{Bmatrix}\,.
        \end{equation}
        Note that this is an orthogonal basis of \(\harmonic{3}{2}\) viewed as a \(\mathbb{C}\)-vector space. Renormalizing then gives us
        \begin{equation}
            \begin{Bmatrix}
                \frac{X_0^2-X_1^2-2\mathrm{i}\,X_0X_1}{2},\\
                X_0X_2-\mathrm{i}\,X_1X_2,\\
                \frac{2\,X_2^2-X_0^2-X_1^2}{\sqrt{6}},\\
                X_0X_2+\mathrm{i}\,X_1X_2,\\
                \frac{X_0^2-X_1^2+2\mathrm{i}\,X_0X_1}{2}
            \end{Bmatrix}\,.
        \end{equation}
        We can then apply the isomorphism to get the set of vectors that span the first eigenspace of \(\rrho{3}{2}\), that is
        \begin{equation}
            \begin{Bmatrix}
                \frac{\ket{00}-\ket{11}-\mathrm{i}\,\ket{01}-\mathrm{i}\,\ket{10}}{2},\\
                \frac{\ket{02}+\ket{20}-\mathrm{i}\,\ket{12}-\mathrm{i}\,\ket{21}}{2},\\
                \frac{2\,\ket{22}-\ket{00}-\ket{11}}{\sqrt{6}},\\
                \frac{\ket{02}+\ket{20}+\mathrm{i}\,\ket{12}+\mathrm{i}\,\ket{21}}{2},\\
                \frac{\ket{00}-\ket{11}+\mathrm{i}\,\ket{01}+\mathrm{i}\,\ket{10}}{2}
            \end{Bmatrix}\,.
        \end{equation}
        In general, one can use the orthogonal basis for spherical harmonics derived in~\cite[Theorem~5.1]{DX13} and transform it into a basis of \(\harmonicd{t}\) using the aforementioned transformation.
    
	\subsection{Proof of Corollary~\ref{cor:trace_distance}}
	\label{app:trace_distance_proof}
        Using the expression for the eigenvalues of \(\rrhodt\) along with their multiplicity, it follows that the trace distance between \(\rrhodt\) and \(\crho\) can be expressed as
        \begin{equation}
            \frac12\left\|\rrhodt-\crho\right\|_1=\sum_k\alpha_k\left|\lambda_k-\binom{d+t-1}{t}^{-1}\right|\,.
        \end{equation}
        The following Corollary shows that under some constraint on \(d\) and \(t\), the trace distance admits a simpler form.
        \tracedistance*
	\begin{proof}
		Let us denote \(\lambda_k\) the \(k\)-th lowest eigenvalue of \(\rrhodt\):
		\begin{equation}
			\lambda_k = \frac{t!(d-2)!!}{(2k)!!(d+2t-2k-2)!!}
		\end{equation}
		and let us denote \(\alpha_k\) its associated multiplicity. Since \(\crho\) can be diagonalized in the same basis as \(\rrhodt\) and has a single non-zero eigenvalue, we have
		\begin{equation}
			\frac12\left\|\rrhodt-\crho\right\|_1=\sum_{k=0}^{K-1}\alpha_k\left(\frac{1}{\binom{d+t-1}{t}}-\lambda_k\right)
		\end{equation}
		with \(K\) being the smallest number such that for all \(k\geqslant K\), we have that \(\lambda_k\geqslant\frac{1}{\binom{d+t-1}{t}}\). If \Cref{eq:trace_distance_condition} holds, then it means that \(K=1\). The trace distance is thus equal to
		\begin{subequations}
			\begin{align}
				&\left[\binom{d+t-1}{d-1}-\binom{d+t-3}{d-1}\right]\left[\frac{1}{\binom{d+t-1}{t}}-\frac{t!(d-2)!!}{(d+2t-2)!!}\right]\\
				={}&\binom{d+t-1}{d-1}\left[1-\frac{t(t-1)}{(d+t-1)(d+t-2)}\right]\left[\frac{1}{\binom{d+t-1}{d-1}}-\frac{t!(d-2)!!}{(d+2t-2)!!}\right]\\
				={}&\left[1-\frac{t(t-1)}{(d+t-1)(d+t-2)}\right]\left[1-\frac{t!(d-2)!!(d+t-1)!}{(d-1)!t!(d+2t-2)!!}\right]\\
				={}&\left[1-\frac{t(t-1)}{(d+t-1)(d+t-2)}\right]\left[1-\frac{(d+t-1)!}{(d-1)!!(d+2t-2)!!}\right]
			\end{align}
		\end{subequations}
		which recovers \Cref{eq:trace_distance} as expected.
	\end{proof}
    \subsubsection{Proof of the asymptotics}
        In this subsection, we want to study the asymptotic regime of \Cref{eq:trace_distance_condition}, which we restate here for convenience
        \[
            \frac{t!(d-2)!!}{2(d+2t-4)!!}\geqslant\binom{d+t-1}{t}^{-1}
        \]
        for \(d\) and \(t\) being two natural numbers.

        We will show that under a suitable condition on \(t\) and \(d\), this inequality holds for sufficiently large \(d\). We will also compute the asymptotics of the trace distance between \(\rrhodt\) and \(\crho\).

        All in all, we show the following Lemma.
        \begin{lemma}
            \label{lem:asymptotics}
            For \(d\) and \(t\) being natural numbers, if \(t<\frac{5+\sqrt{9+8d\ln\left(\frac{d}{2}\right)}}{2}\), then
            \begin{equation}
                \label{eq:trace_distance_condition_restated}
                \frac{t!(d-2)!!}{2(d+2t-4)!!}\geqslant\binom{d+t-1}{t}^{-1}
            \end{equation}
            holds for a sufficiently large \(d\). Furthermore, we have under this condition that
            \begin{equation}
                \frac12\left\|\rrhodt-\crho\right\|_1=\left(1+\bigtheta{\frac{t^2}{d^2}}\right)\left(1-\mathrm{e}^{-\frac{t(t-1)}{2d}+\bigtheta{\frac{t^3}{d^2}}}\right)\,.
            \end{equation}
        \end{lemma}
        \begin{proof}
            First of all, note that \Cref{eq:trace_distance_condition_restated} is equivalent to
            \begin{equation}
    		\label{eq:equivalent_condition_restated}
    		\ln\left(\frac{(d+t-1)!}{(d-1)!!(d+2t-4)!!}\right)\geqslant\ln(2)\,.
    	\end{equation}
            Let us recall that Stirling's approximation can be written as
		\begin{equation}
			\ln(n!)=\left(n+\frac12\right)\ln(n)-n+\frac12\ln(2\pi)+\frac{1}{12n}+\bigtheta{\frac{1}{n^3}}\,.
		\end{equation}
		We thus have
            \begin{equation}
                  \ln((d+t-1)!)=\left(d+t-\frac12\right)\ln(d+t-1)-d-t+1+\frac12\ln(2\pi)+\frac{1}{12(d+t-1)}+\bigtheta{\frac{1}{(d+t)^3}}\,.
            \end{equation}
            Using the Taylor expansion of \(x\mapsto\frac{1}{1+x}\) and the fact that \(t=\smallo{d}\), we get
            \begin{equation}
                \ln((d+t-1)!) =\left(d+t-\frac12\right)\ln(d+t-1)-d-t+1+\frac12\ln(2\pi)+\frac{1}{12d}+\bigtheta{\frac{t}{d^2}}\,.
            \end{equation}
            We can then use the Taylor expansion of \(x\mapsto\ln(1+x)\) to get
            \begin{equation}
                \begin{split}
    			&\ln((d+t-1)!)\\={}&\left(d+t-\frac12\right)\left[\ln(d)+\frac{t-1}{d}-\frac{(t-1)^2}{2d^2}+\bigtheta{\frac{t^3}{d^3}}\right]-d-t+1+\frac12\ln(2\pi)+\frac{1}{12d}+\bigo{\frac{t}{d^2}}
			\end{split}
            \end{equation}
            which we can simplify to
            \begin{equation}
                \ln((d+t-1)!) = \left(d+t-\frac12\right)\ln(d)-d+\frac12\ln(2\pi)+\frac{t(t-1)}{2d}+\frac{1}{12d}+\bigtheta{\frac{t^3}{d^2}}\,.
            \end{equation}
		Similarly, we can apply Stirling's approximation to the double factorial, which is stated as
		\begin{subequations}
			\begin{align}
				\ln((2n)!!)&= n\ln(2)+\left(n+\frac12\right)\ln(n)-n+\frac12\ln(2\pi)+\frac{1}{12n}+\bigtheta{\frac{1}{n^3}}\\
				\ln((2n-1)!!) &= n\ln(n)+\left(n+\frac12\right)\ln(2)-n-\frac{1}{24n}+\bigtheta{\frac{1}{n^3}}\,.
			\end{align}
		\end{subequations}
		We thus have if \(d\) is even
            \begin{equation}
                \ln((d+2t-4)!!) + \ln((d-1)!!) = \ln\left(\left(2\left(\frac{d}{2}+t-2\right)\right)!!\right) + \ln\left(\left(2\frac{d}{2}-1\right)!!\right)\,.
            \end{equation}
            We can now use Stirling's approximation and gather the terms from these expressions to get
            \begin{equation}
                \begin{split}
					&\ln((d+2t-4)!!) + \ln((d-1)!!)\\={}& \left(\frac{d}{2}+t-\frac32\right)\ln\left(d+2t-4\right)-d-t+2+\frac12\ln(2\pi)+\frac{1}{6\left(d+2t-4\right)}+{}\\&\frac{d}{2}\ln\left(d\right)-\frac{1}{12d}+\bigtheta{\frac{1}{d^3}}\,.
				\end{split}
            \end{equation}
            We can then use the Taylor expansions of \(x\mapsto\ln(1+x)\) and \(x\mapsto\frac{1}{1+x}\) to get
            \begin{equation}
                \label{eq:asymptotics_d_even}
                \ln((d+2t-4)!!) + \ln((d-1)!!) = \left(d+t-\frac32\right)\ln(d)-d+\frac12\ln(2\pi)+\frac{t^2-3t+\frac{25}{12}}{d}+\bigtheta{\frac{t^3}{d^2}}\,.
            \end{equation}
		This then yields in this case
		\begin{equation}
			\ln((d+t-1)!)-\ln((d+2t-4)!!)-\ln((d-1)!!) = \ln(d)-\frac{(t-4)(t-1)}{2d}+\bigtheta{\frac{t^3}{d^2}}\,.
		\end{equation}
  
		On the other hand, if \(d\) is odd we have
            \begin{equation}
                \ln((d+2t-4)!!) + \ln((d-1)!!) = \ln\left(\left(2\left(\frac{d-1}{2}+t-1\right)-1\right)!!\right) + \ln\left(\left(2\frac{d-1}{2}\right)!!\right)\,.
            \end{equation}
            We once again use Stirling's approximation and gather the terms to get
            \begin{equation}
                \begin{split}
                    \ln((d+2t-4)!!) + \ln((d-1)!!)
                    ={}&\left(\frac{d-1}{2}+t-1\right)\ln\left(d+2t-3\right)-d-t+2-\frac{1}{12(d+2t-3)}+{}\\&\frac{d}{2}\ln(d-1)+\frac12\ln(2\pi)+\frac{1}{6(d-1)}+\bigtheta{\frac{1}{d^3}}\,.
                \end{split}
            \end{equation}
            Using the Taylor expansion of \(x\mapsto\ln(1-x)\), \(x\mapsto\ln(1+x)\) and \(x\mapsto\frac{1}{1+x}\) then gives us
            \begin{equation}
                \ln((d+2t-4)!!) + \ln((d-1)!!) = \left(d+t-\frac32\right)\ln(d)-d+\frac12\ln(2\pi)+\frac{t^2-3t+\frac{25}{12}}{d}+\bigtheta{\frac{t^3}{d^2}}\,.
            \end{equation}
		Note that this expression is the same as the one in \Cref{eq:asymptotics_d_even}, found in the case where \(d\) is even.
		
		Thus, in all cases, we have
		\begin{equation}
                \label{eq:asymptotics_ln}
			\ln\left(\frac{(d+t-1)!}{(d-1)!!(d+2t-4)!!}\right)=\ln(d)-\frac{(t-4)(t-1)}{2d}+\bigtheta{\frac{t^3}{d^2}}\,.
		\end{equation}
            Thanks to \Cref{eq:asymptotics_ln}, we know that if \(t=\smallo{d^{\frac23}}\), then \Cref{eq:equivalent_condition_restated} can be restated as
            \begin{equation}
                \ln(d)-\frac{(t-1)(t-4)}{2d}+\smallo{1}\geqslant\ln(2)\,.
            \end{equation}
            In particular, if \(d\) is sufficiently large, then this inequality holds whenever
            \begin{subequations}
                \begin{align}
                    &\ln(d)-\frac{(t-1)(t-4)}{2d}>\ln(2)\\
                    \iff&t<\frac{5+\sqrt{9+8d\ln\left(\frac{d}{2}\right)}}{2}\,.
                \end{align}
            \end{subequations}
            
		Furthermore, we know from \Cref{cor:trace_distance} that if this inequality holds, then we have
        \[\frac12\left\|\rrhodt-\crho\right\|_1=\left(1-\frac{t(t-1)}{(d+t-1)(d+t-2)}\right)\left(1-\frac{(d+t-1)!}{(d-1)!!(d+2t-2)!!}\right)\,.\]
        We can then reuse \Cref{eq:asymptotics_ln} to assert that
        \begin{equation}
            \ln\left(\frac{(d+t-1)!}{(d-1)!!(d+2t-2)!!}\right)=\ln(d)-\frac{(t-4)(t-1)}{2d}+\bigtheta{\frac{t^3}{d^2}}-\ln(d+2t-2)\,.
        \end{equation}
        We can then use the Taylor expansion of \(x\mapsto\ln(1+x)\) to get
        \begin{equation}
            \ln\left(\frac{(d+t-1)!}{(d-1)!!(d+2t-2)!!}\right)=\ln(d)-\frac{(t-4)(t-1)}{2d}+\bigo{\frac{t^3}{d^2}}-\ln(d)-\frac{2t-2}{d}+\bigtheta{\frac{t^2}{d^2}}
        \end{equation}
        which can be simplified to
        \begin{equation}
            \ln\left(\frac{(d+t-1)!}{(d-1)!!(d+2t-2)!!}\right) = -\frac{t(t-1)}{2d}+\bigtheta{\frac{t^3}{d^2}}\,.
        \end{equation}
        Finally, applying \(\exp\) to this expression yields
        \begin{equation}
            \frac{(d+t-1)!}{(d-1)!!(d+2t-2)!!} = \mathrm{e}^{-\frac{t(t-1)}{2d}+\bigtheta{\frac{t^3}{d^2}}}\,.
        \end{equation}
        All in all, if \(t<\frac{5+\sqrt{9+8d\ln\left(\frac{d}{2}\right)}}{2}\), then we have
        \begin{equation}
            \frac12\left\|\rrhodt-\crho\right\|_1 = \left(1+\bigtheta{\frac{t^2}{d^2}}\right)\left(1-\mathrm{e}^{-\frac{t(t-1)}{2d}+\bigtheta{\frac{t^3}{d^2}}}\right)\,.
        \end{equation}
        \end{proof}

        \subsection{Exact security analysis of the binary phases ARS}
        \label{app:trace-distance-binary-phases}
        In order to give a more precise asymptotics of \cite{HBK24}'s argument on the number of copies required to test for the imaginarity of a state, we need to know for which regime does the trace distance between the density matrix resulting from sampling \(t\) copies of a binary phase state and the one resulting from sampling \(t\) copies of a Haar random state converges to a given value. In order to do so, we will compute exactly the trace distance between these two density matrices.

        Note that our result only holds whenever \(t\leqslant d\). Fortunately, as we'll see, this covers every interesting regime for pseudorandomness, since this trace distance converges to \(1\) whenever \(t=\omega\!\left(\sqrt{d}\right)\). Extending this result to \(t>d\) using the same technique is certainly possible, but this assumption greatly simplifies the analysis of the problem.

        \begin{lemma}
            \label{lem:trace_distance_binary_phase_ars}
            Let \(d\geqslant2\) and \(t\leqslant d\) be two natural numbers. Let \(\mathcal{F}\) be the set of functions from \([d]\) to \(\{0, 1\}\). Let \(\meanrho{\mathcal{B}_d^t}\) be the density matrix resulting from sampling \(t\) copies of a binary phase state, that is
            \begin{equation}
                \meanrho{\mathcal{B}_d^t} \defed \frac{1}{2^{d}}\sum_{f\in\mathcal{F}}(\selfouter{\psi_f})^{\otimes t}
            \end{equation}
            where \(\ket{\psi_f}\) is defined as
            \begin{equation}
                \ket{\psi_f} \defed \frac{1}{\sqrt{d}}\sum_x(-1)^{f(x)}\ket{x}
            \end{equation}
            for \(f\) being a binary function on \([d]\). We have
            \begin{equation}
                \label{eq:trace_distance_binary_phase_states}
                \frac12\left\|\meanrho{\mathcal{B}_d^t}-\crho\right\|_1 = 1-\frac{1}{\binom{d+t-1}{t}}\sum_{k=0}^{\left\lfloor\frac{t}{2}\right\rfloor}\binom{d}{t-2k}\,.
            \end{equation}
        \end{lemma}
        \begin{proof}
         By a similar argument to the one that was used in \cite[Subsection~4.2.3]{BS19}, we have
        \begin{equation}
            \meanrho{\mathcal{B}_d^t} = \frac{1}{d^t}\sum_{x\sim y}\ketbra{x}{y}
        \end{equation}
        with \(\sim\) being an equivalence relation on \([d]^t\) which is defined as \(x\sim y\) if and only if no element \(z\in[d]\) appears an odd number of times in the concatenation \(x\parallel y\) of \(x\) and \(y\). We thus know via \Cref{lem:eigendecomposition-equivalence-relation} that we can find the eigenvalues of \(\meanrho{\mathcal{B}_d^t}\) and their multiplicities by studying the equivalence classes of \(\sim\). Indeed, in this case, this Lemma asserts that the eigenvalue of \(\meanrho{\mathcal{B}_d^t}\) are exactly the sizes of the equivalence classes of \(\sim\) divided by \(\frac{1}{d^t}\). Let us thus characterize these equivalence classes.

        Suppose an element \(z\in[d]\) appears an odd number of times in \(x\in[d]^t\). Then in order for \(y\in[d]^t\) to be in relation with  \(x\), \(y\) must also contain \(z\) an odd number of times. Putting things differently, two elements \(x\) and \(y\) are in relation if and only if they have in common their elements that appear an odd number of times in them. Thus, an equivalence class is uniquely defined by a set \(\left\{z_1,\cdots,z_k\right\}\in[d]^k\) which defines the elements of \([d]\) that appear an odd number of times in the $t$-tuple. Note that \(k\) necessarily has the same parity as \(t\). Indeed, if \(t\) is odd and \(k\) is even, then there will necessarily be another element \(z_{k+1}\) that must appear in \(x\) and \(y\) an odd number of times to complete the \(t\)-tuple. Similarly, if \(t\) is even, then so must be \(k\).

        Let us write \(t\) as \(t=2T+b\) with \(b\in\{0, 1\}\) being its parity and let us consider an equivalence class uniquely defined by a set of \(2k+b\) elements of \([d]\). We are then interested in the size of this equivalence class. Note that this size does not depend on the individual values of \(z_1,\cdots,z_{2k+b}\) but only on \(k\). In particular, it means that two distinct equivalence classes defined by sets of the same size \(2k+b\) have the same size. It means that if we denote \(\lambda_{k}\) the associated eigenvalue, it has multiplicity at least \(\binom{d}{2k+b}\), since there are \(\binom{d}{2k+b}\) ways of choosing the \(2k+b\) elements of \([d]\) that must have an odd cardinality.

        Furthermore, note that the larger \(k\) and the lower the size of the equivalence class. Indeed, consider an equivalence class defined by a set containing \(2k+2+b\) elements \(\left\{z_1,\cdots,z_{2k+2+b}\right\}\). Consider a \(t\)-tuple within this equivalence class and replace each \(z_{2k+1+b}\) and \(z_{2k+2+b}\) by any element \(z'\in[d]\). Note that this results in an element of the equivalence class defined by \(\left\{z_1,\cdots,z_{2k+b}\right\}\). Indeed, since both \(z_{2k+1+b}\) and \(z_{2k+2+b}\) appear an odd number of times, the sum of their cardinality is necessarily even. As such, for each element in an equivalence class defined by \(2k+b+2\) elements we can find strictly more than one element in the equivalence class defined by \(\left\{z_1,\cdots,z_{2k+b}\right\}\). In particular, the smallest equivalence class is the one defined by having the maximum number of odd cardinality elements, namely \(t\) (recall that we assumed \(t\leqslant d\)). The size of this equivalence class is then the number of permutations of these elements, namely \(t!\). As such, the smallest positive eigenvalue of \(\meanrho{\mathcal{B}_d^t}\) is \(\frac{t!}{d^t}\).

        Now, note that we have
        \begin{equation}
        \rk{\meanrho{\mathcal{B}_d^t}}=\sum_{k=0}^{\left\lfloor\frac{t}{2}\right\rfloor}\binom{d}{t-2k}
        \end{equation}
        which we find by summing the multiplicities of the eigenvalues. In particular, when seen as an operator on \(\symsubspacedt{C}\), \(\meanrho{\mathcal{B}_d^t}\) is not full-rank. As such, the lowest eigenvalue of \(\meanrho{\mathcal{B}_d^t}\) is 0, and its second lowest is \(\frac{t!}{d^t}\). But now, note that we have
        \begin{equation}
            \frac{(d+t-1)!}{(d-1)!}\geqslant d^t
        \end{equation}
        which can be seen by rewriting the left-hand side term as a product of \(t\) terms larger than \(d\). In particular, this means that we have
        \begin{equation}
            \frac{1}{\binom{d+t-1}{t}}\leqslant\frac{t!}{d^t}\,.
        \end{equation}
        Since the second lowest eigenvalue of \(\meanrho{\mathcal{B}_d^t}\) is always larger than the unique eigenvalue of \(\crho\), their trace distance admits a simple form, namely
        \begin{equation}
            \frac12\left\|\meanrho{\mathcal{B}_d^t}-\crho\right\|_1 = \frac{\binom{d+t-1}{t}-\rk{\meanrho{\mathcal{B}_d^t}}}{\binom{d+t-1}{t}}
        \end{equation}
        with \(\binom{d+t-1}{t}-\rk{\meanrho{\mathcal{B}_d^t}}\) being the multiplicity of the 0 eigenvalue of \(\meanrho{\mathcal{B}_d^t}\) viewed as an operator on \(\symsubspacedt{C}\).
        \end{proof}
        Now that we have this expression for the trace distance, we are able to give out its asymptotics. In particular, we have
        \begin{equation}
            \frac12\left\|\meanrho{\mathcal{B}_d^t}-\crho\right\|_1 = \frac{t(t-1)}{d}+\bigtheta{\frac{t^2}{d^4}}\,.
        \end{equation}
        Furthermore, we can use Stirling's approximation to get that if \(t\sim\alpha\sqrt{d}\), then
        \begin{equation}
            \frac12\left\|\meanrho{\mathcal{B}_d^t}-\crho\right\|_1 = 1-\mathrm{e}^{-\alpha^2}+\smallo{1}\,.
        \end{equation}    
    
	\section{Additional lemmas and proofs}
        \begin{lemma}
            \label{lem:expectationY}
            Let \(d\geqslant2\) be a natural number. Let \(n_0,\cdots,n_{d-1}\) be natural numbers and let us denote their sum by \(N\). Let \(W\in\RR^d\) be distributed according to a standard centered multivariate normal law \(\mathcal{N}\left(0,\id_d\right)\). We have
            \begin{equation}
                \expect{\prod_{i=0}^{d-1}\frac{W_i^{n_i}}{\|W\|^{n_i}}} = \begin{cases}
                    \frac{(d-2)!!}{(d+N-2)!!}\prod\limits_{i=0}^{d-1}\left(n_i-1\right)!!&\text{ if all }n_i\text{ are even,}\\
                    0&\text{otherwise.}
                \end{cases}
            \end{equation}
        \end{lemma}
        \begin{proof}
            We'll use the well-known fact that the magnitude \(\|W\|\) and the direction \(\frac{W}{\|W\|}\) of a gaussian vector are independent random variables. We then have on the one hand
            \begin{equation}
                \expect{\prod_{i=0}^{d-1}W_i^{n_i}} = \prod_{i=0}^{d-1}\expect{W_i^{n_i}}
            \end{equation}
            since the \(\left\{W_i\right\}_{i}\) are independent. We then have that if one of the \(n_i\) is odd, this expectation is equal to \(0\). Otherwise we have
            \begin{equation}
                \expect{\prod_{i=0}^{d-1}W_i^{n_i}} = \prod_{i=0}^{d-1}\left(n_i-1\right)!!\,.
            \end{equation}
            On the other hand, we have
            \begin{equation}
                \expect{\prod_{i=0}^{d-1}W_i^{n_i}} = \expect{\|W\|^N\prod_{i=0}^{d-1}\frac{W_i^{n_i}}{\|W\|^n_i}} = \expect{\|W\|^N}\expect{\prod_{i=0}^{d-1}\frac{W_i^{n_i}}{\|W\|^n_i}}\,.
            \end{equation}
            Thus, we know that the expectation that we want to compute will be nil if at least one of the \(n_i\) is odd. Let us thus assume that all the \(n_i\) are even. This means that \(N\) is also even, which means that \(\|W\|^N\) follows a \(\chi^2\left(\frac{N}{2}\right)\) law. This gives us
            \begin{equation}
                \expect{\|W\|^{N}} = \frac{\left(d+N-2\right)!!}{\left(d-2\right)!!}
            \end{equation}
            from which the result follows.
        \end{proof}
	\begin{lemma}
		\label{lem:doublefactorialbinomialformula}
		Let \((n, p, q)\in\mathbb{N}\). We have
		\begin{equation}
			\sum_{k=0}^n\binom{n}{k}(2n-2k-1+q)!!(2k-1+p)!!=\frac{(p-1)!!(q-1)!!}{(p+q)!!}(2n+p+q)!!\,.
		\end{equation}
	\end{lemma}
	\begin{proof}
		We prove this by induction on \(n\). We denote \(u_{n,p,q}\) this sum. First, note that this is true for \(n=0\). Let us thus assume that this assertion is true for \(n\in\NN\) and all \(p, q\in\NN\) and let us show that it also holds for \(n+1\) and all \(p,q\in\NN\). We have
        \begin{equation}
            u_{n+1,p,q} = \sum_{k=0}^{n+1}\binom{n+1}{k}(2n-2k+1+q)!!(2k-1+p)!!\,.
        \end{equation}

        Let us first remove the first and last terms out of the sum to get
        \begin{equation}
            u_{n+1,p,q}=(2n+1+q)!!(p-1)!! + (q-1)!!(2n+1+p)!! + \sum_{k=1}^{n}\binom{n+1}{k}(2n-2k+1+q)!!(2k-1+p)!!\,.
        \end{equation}
        We can then split \(\binom{n+1}{k}\) into \(\binom{n}{k}+\binom{n}{k-1}\) to get
        \begin{equation}
            \begin{split}
                u_{n+1,p,q}={}&(2n+1+q)!!(p-1)!! + \sum_{k=1}^{n}\binom{n}{k}(2n-2k+1+q)!!(2k-1+p)!! +{}\\&(q-1)!!(2n+1+p)!! + \sum_{k=1}^{n}\binom{n}{k-1}(2n-2k+1+q)!!(2k-1+p)!!\,.
            \end{split}
        \end{equation}
        We then change the index of the second sum to get
        \begin{equation}
            \begin{split}
                u_{n+1,p,q}={}&(2n+1+q)!!(p-1)!! + \sum_{k=1}^{n}\binom{n}{k}(2n-2k+1+q)!!(2k-1+p)!! +{}\\&(q-1)!!(2n+1+p)!! + \sum_{k=0}^{n-1}\binom{n}{k}(2n-2k-1+q)!!(2k+1+p)!!\,.
            \end{split}
        \end{equation}
        We are now able to merge into these two sums the terms that we previously left out to get
        \begin{equation}
            u_{n+1,p,q}=\sum_{k=0}^{n}\binom{n}{k}(2n-2k+1+q)!!(2k-1+p)!! +\sum_{k=0}^{n}\binom{n}{k}(2n-2k-1+q)!!(2k+1+p)!!
        \end{equation}
        which we identify as
        \begin{equation}
            u_{n+1,p,q} = u_{n,p,q+2} + u_{n,p+2,q}\,.
        \end{equation}
        We thus have by our induction assumption
        \begin{equation}
            u_{n+1,p,q} = \frac{(p-1)!!(q+1)!!}{(p+q+2)!!}(2n+p+q+2)!! + \frac{(p+1)!!(q-1)!!}{(p+q+2)!!}(2n+p+q+2)!!\,.
        \end{equation}
        We then factorize this expression to get
        \begin{equation}
            u_{n+1,p,q} = \frac{(2n+p+q+2)!!}{(p+q+2)!!}(p-1)!!(q-1)!!(p+1+q+1)
        \end{equation}
        which we can finally simplify to
        \begin{equation}
            u_{n+1,p,q} = \frac{(2(n+1)+p+q)!!}{(p+q)!!}(p-1)!!(q-1)!!
        \end{equation}
        which concludes the proof.
	\end{proof}
	\begin{corollary}
		\label{cor:multinome_double_factorial}
		Let \(d\geqslant2\), \(k_0\in\mathbb{N}\), \(p\in\mathbb{N}\) and \(q\in\mathbb{N}\) be four natural numbers. We have
		\begin{equation}
			\begin{aligned}
				&\sum_{k_{d-1}\leqslant\cdots\leqslant k_1=0}^{k_0}\left[\prod_{j=0}^{d-2}\binom{k_j}{k_{j+1}}\left(2k_j-2k_{j+1}+q-1\right)!!\right]\left(2k_{d-1}+p-1\right)!!\\
				={}&\frac{\left(2k_0+p+(d-1)(q+1)-1\right)!!(p-1)!!\left[(q-1)!!\right]^{d-1}}{(p+(d-1)(q+1)-1)!!}\,.
			\end{aligned}
		\end{equation}
	\end{corollary}
	\begin{proof}
		We prove this by induction on \(d\). The case \(d=2\) reduces to the previous lemma and is thus true. Let us assume that this proposition is true for some \(d\), and let's prove that it's also true for \(d+1\). Let us denote
        \begin{equation}
            u_{d, k_0,p,q} = \sum_{k_{d-1}\leqslant\cdots\leqslant k_1=0}^{k_0}\left[\prod_{j=0}^{d-2}\binom{k_j}{k_{j+1}}\left(2k_j-2k_{j+1}+q-1\right)!!\right]\left(2k_{d-1}+p-1\right)!!
        \end{equation}
        for all \(d\geqslant2\), \(k_0\in\NN\), \(p\in\NN\) and \(q\in\NN\). In particular, we have
        \begin{equation}
            u_{d+1, k_0,p,q} = \sum_{k_{d}\leqslant\cdots\leqslant k_1=0}^{k_0}\left[\prod_{j=0}^{d-1}\binom{k_j}{k_{j+1}}\left(2k_j-2k_{j+1}+q-1\right)!!\right]\left(2k_{d}+p-1\right)!!\,.
        \end{equation}
        By summing over the terms that depend only on the \(k_d\) index, we find that we have
        \begin{equation}
            u_{d+1,k_0,p,q}=\sum_{k_{d-1}\leqslant\cdots\leqslant k_1=0}^{k_0}\left[\prod_{j=0}^{d-2}\binom{k_j}{k_{j+1}}\left(2k_j-2k_{j+1}+q-1\right)!!\right]u_{2,k_{d-1},p, q}\,.
        \end{equation}
        We then have by \Cref{lem:doublefactorialbinomialformula}
        \begin{equation}
            \begin{split}
                &u_{d+1,k_0,p,q}\\
                ={}&\frac{(p-1)!!(q-1)!!}{(p+q)!!}\sum_{k_{d-1}\leqslant\cdots\leqslant k_1=0}^{k_0}\left[\prod_{j=0}^{d-2}\binom{k_j}{k_{j+1}}\left(2k_j-2k_{j+1}+q-1\right)!!\right]\left(2k_{d-1}+p+q\right)!!
            \end{split}
        \end{equation}
        which we can identify as
        \begin{equation}
            u_{d+1,k_0,p,q}=\frac{(p-1)!!(q-1)!!}{(p+q)!!}u_{d,k_0,p+q+1,q}\,.
        \end{equation}
        We thus have by our induction assumption
        \begin{equation}
            u_{d+1,k_0,p,q}=\frac{(p-1)!!(q-1)!!\left[2k_0+p+q+1+(d-1)(q+1)-1\right]!!(p+q)!!\left[(q-1)!!\right]^{d-1}}{(p+q)!![p+q+1+(d-1)(q+1)-1]!!}
        \end{equation}
        which we can finally simplify to
        \begin{equation}
            u_{d+1,k_0,p,q}=\frac{(p-1)!!\left[2k_0+p+d(q+1)-1\right]!!\left[(q-1)!!\right]^{d}}{[p+d(q+1)-1]!!}
        \end{equation}
        which concludes the proof.
	\end{proof}
        \begin{lemma}
            \label{lem:eigendecomposition-equivalence-relation}
            Let \(d\geqslant2\) and \(t\geqslant1\) be two natural numbers. Let \(\sim\) be an equivalence relation on \([d]^t\). For an element \(x\) of \([d]^t\), we define \(\mathcal{C}(x)\) to be its equivalence class with respect to \(\sim\) and we define:
            \begin{equation}
                \ket{\mathcal{C}(x)}=\sum_{y\sim x}\ket{y}\,.
            \end{equation}
            Let \(\rho\) be a density matrix that can be written as:
            \begin{equation}
                \rho=\sum_{x\in[d]^t}\alpha_{\mathcal{C}(x)}\sum_{\substack{y\in[d]^t\\y\sim x}}\ketbra{x}{y}
            \end{equation}
            with \(\left\{\alpha_{\mathcal{C}(x)}\right\}_{x\in[d]^t}\) being arbitrary positive numbers that are constant across equivalence classes.
            We then have:
            \begin{equation}
                \rho=\sum_{z\in[d]^t_{/\sim}}\alpha_{\mathcal{C}(z)}|\mathcal{C}(z)|\,\proj{\mathcal{C}(z)}\,.
            \end{equation}
            In particular, since \(\inner{\mathcal{C}(x)}{\mathcal{C}(y)}\neq0\iff \mathcal{C}(x)=\mathcal{C}(y)\), this is the spectral decomposition of \(\rho\).
        \end{lemma}
        \begin{proof}
            Let \(z\in[d]^t\). Let us show that \(\ket{\mathcal{C}(z)}\) s an eigenvector of \(\rho\). We have by definition of the different terms
            \begin{equation}
                \rho\ket{\mathcal{C}(z)} = \sum_{x\in[d]^t}\alpha_{\mathcal{C}(x)}\sum_{\substack{y\in[d]^t\\y\sim x}}\ketbra{x}{y}\sum_{w\sim z}\ket{w}\,.
            \end{equation}
            Since \(\inner{y}{w}=\delta_{y,w}\), this simplifies to
            \begin{equation}
                \rho\ket{\mathcal{C}(z)} = \sum_{x\in[d]^t}\alpha_{\mathcal{C}(x)}\sum_{\substack{y\in[d]^t\\y\sim x\\y\sim z}}\ket{x}\,.
            \end{equation}
            Now, note that the summand is independent of the summation index \(y\), that is
            \begin{equation}
                \rho\ket{\mathcal{C}(z)} = \sum_{\substack{x\in[d]^t\\x\sim z}}\alpha_{\mathcal{C}(x)}\,\ket{x}\sum_{\substack{y\in[d]^t\\y\sim x}}1\,.
            \end{equation}
            The rightmost term simply counts the number of elements that are in relation with a given \(x\in[d]^t\), which is the size of \(\mathcal{C}(x)\) by definition. We thus have
            \begin{equation}
                \rho\ket{\mathcal{C}(z)} = \sum_{\substack{x\in[d]^t\\x\sim z}}\alpha_{\mathcal{C}(x)}|\mathcal{C}(x)|\,\ket{x}\,.
            \end{equation}
            Now, note that by assumption the coefficients in the summand are constant across a given equivalence class, which means that we have
            \begin{equation}
                \rho\ket{\mathcal{C}(z)} = \alpha_{\mathcal{C}(z)}|\mathcal{C}(z)|\,\sum_{\substack{x\in[d]^t\\x\sim z}}\ket{x}\,.
            \end{equation}
            Finally, note that the rightmost sum is now \(\ket{\mathcal{C}(z)}\) by definition, which finally gives us
            \begin{equation}
                \rho\ket{\mathcal{C}(z)} = \alpha_{\mathcal{C}(z)}|\mathcal{C}(z)|\,\ket{\mathcal{C}(z)}\,.
            \end{equation}
            Thus, for any \(z\), \(\ket{\mathcal{C}(z)}\) is an eigenvector of \(\rho\) with its eigenvalue being \(\alpha_{\mathcal{C}(z)}|\mathcal{C}(z)|\). In order to consider a single element of each equivalence class, we can quotient \([d]^t\) by \(\sim\). Finally, note that there is no other eigenvectors associated to a non-zero eigenvalue, since:
            \begin{subequations}
                \begin{align}
                    \tr{\rho}=\sum_{z\in[d]^t}\alpha_{\mathcal{C}(z)}=\sum_{z\in[d]^t_{/\sim}}\alpha_{\mathcal{C}(z)}|\mathcal{C}(z)|\,.
                \end{align}
            \end{subequations}
        \end{proof}
        Before introducing the next Lemma, we recall the definition and some properties of the Beta distribution.
        
        A random variable \(X\) follows a \(\mathrm{Beta}(\alpha, \beta)\) distribution with \(\alpha\) and \(\beta\) being two positive numbers if its density function is given by
        \begin{equation}
            f_{\mathrm{Beta}(\alpha,\beta)}(x)=\begin{cases}\frac{\Gamma(\alpha+\beta)}{\Gamma(\alpha)\Gamma(\beta)}x^{\alpha-1}\,(1-x)^{\beta-1}&\text{if }x\in(0, 1)\\0&\text{otherwise}\end{cases}
        \end{equation}
        with \(\Gamma\) being the Gamma function. In particular, note that we have
        \begin{equation}
            \label{eq:beta_integration_identity}
            \int_0^1x^{\alpha-1}(1-x)^{\beta-1}\,\mathrm{d}x=\frac{\Gamma(\alpha)\Gamma(\beta)}{\Gamma(\alpha+\beta)}
        \end{equation}
        for all positive numbers \(\alpha\) and \(\beta\). If \(X\) is distributed according to a \(\mathrm{Beta}(\alpha,\beta)\) law, then its moments are given by
        \begin{equation}
            \label{eq:beta_moments}
            \expect{X^n} = \prod_{k=0}^{n-1}\frac{\alpha+k}{\alpha+\beta+k}
        \end{equation}
        for \(n\in\NN\).

        Finally, we'll make use of the fact that if \(A\) follows a \(\chi^2(p)\) law and \(B\) follows a \(\chi^2(q)\) law and \(A\) and \(B\) are independent, then \(\frac{A}{A+B}\) follows a \(\mathrm{Beta}\left(\frac{p}{2},\frac{q}{2}\right)\) law.
        
        We are now ready to state the following Lemma on the distribution of imaginarity of Haar random states.
        \begin{lemma}
            \label{lem:distribution-of-imaginarity}
            Let \(d\geqslant2\) be a natural number. Recall that the imaginarity of a pure state \(\ket{\psi}\in\complexvecs{d}\) is given by
            \begin{equation}
                \mathcal{I}(\ket{\psi}) = 1-\left|\inner{\overline{\psi}}{\psi}\right|^2\,.
            \end{equation}
            If \(\ket{\psi}\) is distributed according to the Haar measure on \(\complexvecs{d}\), then \(\mathcal{I}(\ket{\psi})\) follows a \(\mathrm{Beta}\left(\frac{d-1}{2}, 1\right)\) law.
        \end{lemma}
        \begin{proof}
            Since \(\ket{\psi}\) is Haar random, we can write it as
            \begin{equation}
                \ket{\psi} = \frac{A+\mathrm{i}\,B}{\sqrt{\|A\|^2+\|B\|^2}}
            \end{equation}
            with \(A\) and \(B\) being independently distributed according to a central standard normal law \(\mathcal{N}\left(0,\id_d\right)\), as per~\cite[Theorem~1]{Mez07}.
            We have
            \begin{equation}
                \inner{\overline{\psi}}{\psi} = \frac{\|A\|^2-\|B\|^2+2\mathrm{i}\inner{A}{B}}{\|A\|^2+\|B\|^2}\,.
            \end{equation}
            This gives us
            \begin{equation}
                \left|\inner{\overline{\psi}}{\psi}\right|^2 = \frac{\left(\|A\|^2-\|B\|^2\right)^2+4\inner{A}{B}^2}{\left(\|A\|^2+\|B\|^2\right)^2}
            \end{equation}
            which finally gives us
            \begin{subequations}
                \begin{align}
                    1-\left|\inner{\overline{\psi}}{\psi}\right|^2 &= \frac{\left(\|A\|^2+\|B\|^2\right)^2-\left(\|A\|^2-\|B\|^2\right)^2-4\inner{A}{B}^2}{\left(\|A\|^2+\|B\|^2\right)^2}\\
                    &= \frac{4\|A\|^2\|B\|^2-4\inner{A}{B}^2}{\left(\|A\|^2+\|B\|^2\right)^2}\\
                    &= \frac{4\|A\|^2\|B\|^2}{\left(\|A\|^2+\|B\|^2\right)^2}\left(1-\inner{\frac{A}{\|A\|}}{\frac{B}{\|B\|}}^2\right)\,.
                \end{align}
            \end{subequations}
            We can rewrite this expression as
            \begin{equation}
                \mathcal{I}(\ket{\psi}) = 4\left(\frac{\|A\|^2}{\|A\|^2+\|B\|^2}\right)\left(1-\frac{\|A\|^2}{\|A\|^2+\|B\|^2}\right)\left(1-\inner{\frac{A}{\|A\|}}{\frac{B}{\|B\|}}^2\right)\,.
            \end{equation}
            Since $A$ is a standard normal gaussian variable, we have that \(\|A\|\) and \(\frac{A}{\|A\|}\) are independent, and similarly for \(B\). As such, we have that \(\frac{\|A\|^2}{\|A\|^2+\|B\|^2}\) and \(\inner{\frac{A}{\|A\|}}{\frac{B}{\|B\|}}^2\) are independent.

            Furthermore, on the one hand, by the rotational invariance of the gaussian distribution, we have that \(\inner{\frac{A}{\|A\|}}{\frac{B}{\|B\|}}^2\) has the same distribution as \(\inner{\frac{A}{\|A\|}}{0}^2\), which we can write as \(\frac{A_0^2}{A_0^2+\sum\limits_{i=1}^{d-1}A_i^2}\). Since \(A_0^2\sim\chi^2(1)\) and \(\sum\limits_{i=1}^{d-1}A_i^2\sim\chi^2(d-1)\) are independent, we have that this variable follows a \(\mathrm{Beta}\left(\frac12,\frac{d-1}{2}\right)\) law. We thus have that \(1-\inner{\frac{A}{\|A\|}}{\frac{B}{\|B\|}}^2\) follows a \(\mathrm{Beta}\left(\frac{d-1}{2},\frac12\right)\) law. In particular, its moments are given by
            \begin{equation}
                \expect{\left(1-\inner{\frac{A}{\|A\|}}{\frac{B}{\|B\|}}^2\right)^n} =\prod_{k=0}^{n-1}\frac{\frac{d-1}{2}+k}{\frac{d}{2}+k}= \frac{(d+2n-3)!!(d-2)!!}{(d-3)!!(d+2n-2)!!}
            \end{equation}
            for \(n\in\NN\) as per~\Cref{eq:beta_moments}.

            On the other hand, by a similar argument, we have that \(\frac{\|A\|^2}{\|A\|^2+\|B\|^2}\) follows a \(\mathrm{Beta}\left(\frac{d}{2},\frac{d}{2}\right)\) law. We thus have
            \begin{equation}
                \expect{\left(\left(\frac{\|A\|^2}{\|A\|^2+\|B\|^2}\right)\left(1-\frac{\|A\|^2}{\|A\|^2+\|B\|^2}\right)\right)^n}=\frac{\Gamma(d)}{\Gamma^2\left(\frac{d}{2}\right)}\int_{0}^1x^n(1-x)^nx^{\frac{d}{2}-1}(1-x)^{\frac{d}{2}-1}\,\mathrm{d}x
            \end{equation}
            for \(n\in\NN\). This simplifies to
            \begin{equation}
                \expect{\left(\left(\frac{\|A\|^2}{\|A\|^2+\|B\|^2}\right)\left(1-\frac{\|A\|^2}{\|A\|^2+\|B\|^2}\right)\right)^n}=\frac{\Gamma(d)\Gamma^2\left(n+\frac{d}{2}\right)}{\Gamma^2\left(\frac{d}{2}\right)\Gamma(2n+d)}
            \end{equation}
            as per~\Cref{eq:beta_integration_identity}. Using the \(\Gamma(z+1)=z\Gamma(z)\) identity, we then find
            \begin{equation}
                \expect{\left(\left(\frac{\|A\|^2}{\|A\|^2+\|B\|^2}\right)\left(1-\frac{\|A\|^2}{\|A\|^2+\|B\|^2}\right)\right)^n}=\frac{(d-1)![(2n+d-2)!!]^2}{4^n(2n+d-1)![(d-2)!!]^2}\,.
            \end{equation}

            Now, since the two variables we computed the moments of are independent, the moments of their product are given by the product of their moments. That is, we have for \(n\in\NN\)
            \begin{equation}
                \expect{\mathcal{I}^n(\ket{\psi})} = \frac{(d-1)!(2n+d-2)!!(d+2n-3)!!}{(2n+d-1)!(d-2)!!(d-3)!!}=\frac{(d-1)!(2n+d-2)!}{(d-2)!(2n+d-1)!}=\frac{d-1}{2n+d-1}\,.
            \end{equation}
            Thus, the moments of \(\mathcal{I}(\ket{\psi})\) match with the moments of the \(\mathrm{Beta}\left(\frac{d-1}{2},1\right)\) law.

            Since \(\mathcal{I}(\ket{\psi})\) is defined over the bounded interval \([0, 1]\), its moments completely define its distribution as per the Hausdorff moment problem, which concludes the proof.
        \end{proof}
\end{document}